\documentclass{article}

\usepackage[preprint]{neurips_2026}

\usepackage[utf8]{inputenc} % allow utf-8 input
\usepackage[T1]{fontenc}    % use 8-bit T1 fonts
\usepackage{hyperref}       % hyperlinks
\usepackage{url}            % simple URL typesetting
\usepackage{booktabs}       % professional-quality tables
\usepackage{amsfonts}       % blackboard math symbols
\usepackage{nicefrac}       % compact symbols for 1/2, etc.
\usepackage{microtype}      % microtypography
\usepackage{xcolor}         % colors
\usepackage{amsmath}
\usepackage{amsthm}
\usepackage{thmtools}
\usepackage{bm}
\usepackage{wrapfig}
\usepackage{graphicx}
\usepackage{xcolor}
\usepackage{stfloats} % For figure* positioning
\usepackage{multirow}
\usepackage[table]{xcolor}
\usepackage{subcaption}
\usepackage{enumitem}

\definecolor{lightgray}{gray}{0.92}  % 定义浅灰色

\newcommand{\name}{\text{LERA}}

\newtheorem{theorem}{Theorem}[section]

\theoremstyle{definition}

\theoremstyle{remark}
\newtheorem{remark}[theorem]{Remark}

\usepackage{booktabs}
\usepackage{tabularx}
\usepackage{array}
\usepackage{xcolor}
\usepackage[most]{tcolorbox}
\usepackage{enumitem}
\usepackage{amssymb}

\newcolumntype{Y}{>{\raggedright\arraybackslash}X}

\newtcolorbox{querybox}[2][]{
    enhanced,
    breakable,
    colback=gray!4,
    colframe=gray!45,
    boxrule=0.5pt,
    arc=1.5pt,
    left=6pt,
    right=6pt,
    top=5pt,
    bottom=5pt,
    title={#2},
    fonttitle=\bfseries,
    coltitle=black,
    colbacktitle=gray!12,
    #1
}

\newtcolorbox{promptbox}[2][]{
    enhanced,
    breakable,
    colback=gray!3,
    colframe=gray!55,
    boxrule=0.5pt,
    arc=1.5pt,
    left=7pt,
    right=7pt,
    top=6pt,
    bottom=6pt,
    title={#2},
    fonttitle=\bfseries,
    coltitle=black,
    colbacktitle=gray!15,
    #1
}

\newtcolorbox{responsebox}[2][]{
    enhanced,
    breakable,
    colback=gray!2,
    colframe=gray!35,
    boxrule=0.4pt,
    arc=1pt,
    left=6pt,
    right=6pt,
    top=5pt,
    bottom=5pt,
    title={#2},
    fonttitle=\itshape,
    coltitle=black,
    colbacktitle=gray!10,
    #1
}

\title{LERA: LLM-Enhanced RAG for Ad Auction in Generative Chatbots}

\author{%
  \textbf{Haoran Sun\textsuperscript{1},
  Xinrui Song\textsuperscript{3},
  Xinyu Zhang\textsuperscript{2},
  Zhaohua Chen\textsuperscript{1}}, \\
  \textbf{Xu Chu\textsuperscript{1},
  Zhilin Zhang\textsuperscript{2},
  Chuan Yu\textsuperscript{2},
  Jian Xu\textsuperscript{2},
  Bo Zheng\textsuperscript{2},
  Xiaotie Deng\textsuperscript{1}} \\
  \textsuperscript{1}Peking University
  \textsuperscript{2}Alibaba Group
  \textsuperscript{3}Shandong University \\
  \texttt{sunhaoran0301@stu.pku.edu.cn} \\
  \texttt{\{chenzhaohua, chu\_xu, xiaotie\}@pku.edu.cn} \\
  \texttt{\{zxy479640, zhangzhilin.pt, yuchuan.yc, xiyu.xj, bozheng\}@alibaba-inc.com}
}

\begin{document}

\maketitle

\begin{abstract}
The integration of advertising auction mechanisms into large language model (LLM)-based chatbots presents a significant opportunity for commercialization, yet poses unique challenges in balancing relevance, efficiency, and user experience. Recently, Feizi et al.~\citep{feizi2023online} and Hajiaghayi et al.~\citep{hajiaghayi2024ad} outlined a retrieve-then-generate paradigm that decouples retrieval and generation, offering lightweight ad insertion and payment determination. However, current retrieval relies solely on text embedding similarity, which may lead to commercial misinterpretation and issues such as repetitive insertions. In this paper, we propose LERA, a two-stage retrieve-then-generate auction framework tailored for LLM chatbots. In the first stage, embedding-based coarse filtering pre-selects a small set of candidate advertisers. In the second stage, the LLM itself is queried with a carefully designed prompt to produce logits over candidates, which serve as refined organic relevance scores. These scores are combined with bids, and a critical-value payment rule accounts for both the coarse-filtering and fine-ranking thresholds, ensuring truthfulness for utility-maximizing advertisers. The framework naturally extends to multiple ad insertions within dynamic dialogue flows and long responses. Experiments on a synthetic advertiser-query benchmark show that LERA substantially improves ad selection accuracy and insertion diversity while incurring only controllable latency overhead. 
\end{abstract}
\section{Introduction}\label{sec:intro}

With the rapid progress of foundation models, large language models (LLMs) are increasingly serving as conversational gateways through which users access information, products, and services~\citep{team2023gemini,qwen3technicalreport,deepseekai2025deepseekv32}.
This shift creates new opportunities for monetizing generative systems beyond subscription and token-based pricing.
A central challenge is how to integrate advertising into LLM-based chatbots in a way that is contextually appropriate for users, commercially meaningful for advertisers, and compatible with auction-based pricing~\citep{feizi2023online,tang2024genai}.
Unlike sponsored search or display advertising, where ads are placed in predefined slots, chatbot advertising must decide whether and how an advertisement should be woven into a generated response.
This makes the auction problem inseparable from language understanding and generation.

A recent line of work studies advertising auctions for generative AI systems, which can be broadly categorized into two design paradigms: \textit{retrieve-then-generate} and \textit{generate-then-retrieve}. 
Retrieve-then-generate methods, such as the RAG-based auction of Hajiaghayi et al.~\citep{hajiaghayi2024ad}, perform lightweight retrieval and insert ads before generation, offering high efficiency but relying on embedding similarity that may not fully capture commercial intent or avoid repetitive insertions. 
Generate-then-retrieve approaches leverage the LLM's contextual understanding to produce response-level signals or advertiser-specific distributions~\citep{feizi2023online,mordo2024sponsored,soumalias2024truthful,duetting2024mechanism}, but they often require multiple LLM calls per query or advertiser-specific scoring, incurring significant latency and inference cost. 
Placeholder-based methods instead generate responses with designated ad slots and fill them afterward~\citep{balseiro2025position,xu2026ad}, potentially compromising fluency because off-the-shelf LLMs are not generally trained to produce natural responses around explicit placeholders.

At the same time, directly using LLMs as ad selectors is not by itself a complete solution. 
Although LLMs have strong contextual understanding and have been used for ranking and recommendation~\citep{yang2023against,li2024leadre,reisenbichler2025applying}, incorporating LLM judgments into an auction requires a payment rule that remains compatible with strategic bidding. 
Moreover, scoring all advertisers with LLM logits can be prohibitively expensive under practical high-throughput inference engines.

These trade-offs motivate \name, a two-stage retrieve-then-generate auction framework for ad insertion in LLM-based chatbots. 
The central idea is to preserve the efficiency of retrieval-based allocation while using the LLM's contextual judgment only where it is most valuable: on a small candidate set. 
As is shown in Figure~\ref{fig:framework}, in the first stage, \name\ applies embedding-based coarse filtering to preselect candidate advertisers from a large pool. 
In the second stage, it presents these candidates, together with a no-insertion option, to the LLM and extracts logits over candidate labels as refined organic relevance scores. These scores are combined with bids to determine the winner, and a critical-value payment rule accounts for both the coarse-filtering and fine-ranking thresholds. 
Thus, although using LLM logits for ranking is common in broader retrieval and recommendation settings, \name\ makes such scoring practical for real-time chatbot ad auctions by embedding it inside a throughput-aware two-stage mechanism. 
Crucially, in high-throughput inference engines like SGLang, logit extraction relies on parallel pre-filling across all candidates. By restricting this computation to only the small set of first-stage winners, \name\ significantly improves throughput while preserving accuracy.
Experiments on single-ad and multi-ad insertion tasks show that \name\ improves ad selection quality and dialogue-level satisfaction with only controllable latency overhead.

Our contributions are summarized as follows:
    \begin{itemize}
    \item We propose \name, a context-aware ad auction framework for LLM chatbots that combines embedding-based coarse filtering with LLM-logit-based fine ranking.
    \item We design a critical-value payment rule tailored to the two-stage allocation process, accounting for both candidate-set entry and final-stage competition.
    \item Qualitative and quantitative experiments show that \name\ better captures nuanced commercial intent, improves insertion diversity, and achieves a favorable accuracy--efficiency trade-off compared with embedding-only and LLM-only alternatives.
\end{itemize}

The remainder of this paper is organized as follows. 
Section~\ref{sec:related_work} provides a concise overview of related work, with a more detailed discussion deferred to Appendix~\ref{sec:related}. 
Section~\ref{sec:pre} introduces the auction setting and the decoupled allocation paradigms for LLM-based advertising. 
Section~\ref{sec:methods} presents the proposed two-stage scoring mechanism, payment rule, and multi-ad extension. 
We conduct extensive experiments to evaluate \name\ against different baselines in Section~\ref{sec:experiment} and conclude in Section~\ref{sec:conclusion}.

\begin{figure*}[t]
    \centering
    \includegraphics[width=0.95\textwidth, trim= 30 30 30 30, clip]{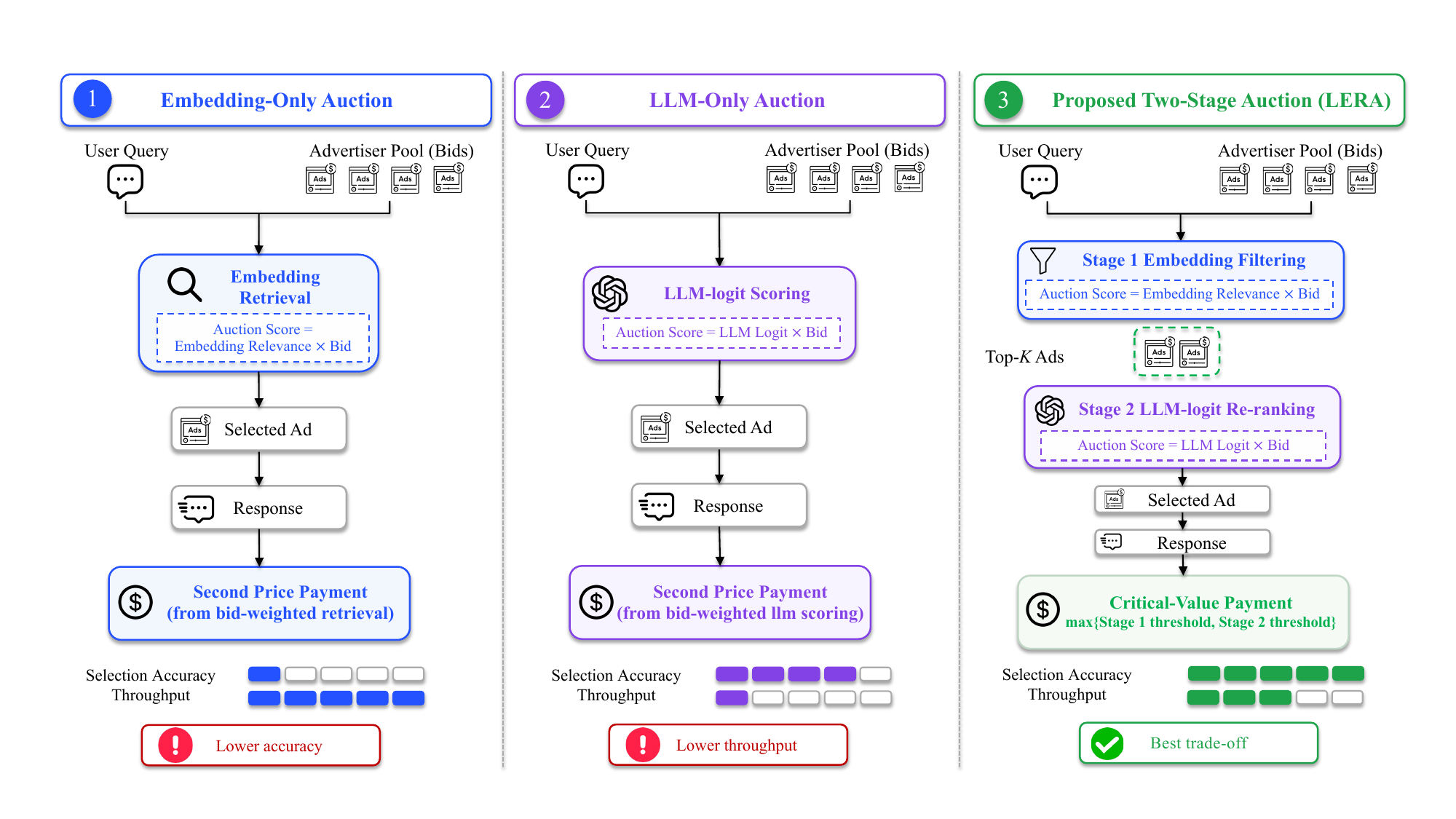}
    \caption{Comparison of ad selection strategies for LLM-based chatbots. Each panel shows a different approach: \textbf{Embedding-only} uses only embedding similarity weighted by bids; \textbf{LLM-Only} scores all advertisers using LLM logits combined with bids; \textbf{\name{}} first filters candidates with embedding similarity and then applies LLM-logit scoring on the top-K set. Payment annotations indicate the auction mechanism used in each method, and the bars illustrate relative selection accuracy and system throughput.}
    \label{fig:framework}
\end{figure*}
\section{Related Work}\label{sec:related_work}

Our framework lies at the intersection of LLM-mediated advertising, retrieval-augmented generation, LLM-based recommendation, and externality-aware auction design. We briefly review these areas and highlight our differences; an extended discussion is provided in Appendix~\ref{sec:related}.

\textbf{Ad Auctions in LLMs.}
Existing frameworks can generally be categorized into generate-then-retrieve and retrieve-then-generate paradigms. Generate-then-retrieve methods auction the right to speak or insert ads based on generated probabilities or responses~\citep{duetting2024mechanism,mordo2024sponsored,soumalias2024truthful}, often requiring multiple inference passes. Retrieve-then-generate methods, such as Hajiaghayi et al.~\cite{hajiaghayi2024ad}, treat advertisers as retrievable documents and are more efficient, but they rely on embedding similarity that may miss nuanced commercial intent. Training-based approaches~\citep{zhao2025llm} embed bid monotonicity directly into the LLM, but are costly to retrain. Our work stays within the efficient retrieve-then-generate paradigm and introduces a modular two-stage scoring that uses LLM judgments only on a shortlist, balancing accuracy and throughput.

\textbf{Retrieval-Augmented Generation.}
Standard RAG pipelines maximize semantic fidelity to a query~\citep{lewis2020retrieval,gao2023retrieval} and have been adapted for recommendation~\citep{csakar2025maximizing}. Ad insertion, however, demands more than topical similarity: it must respect commercial intent, avoid repetitive insertions across a dialogue, and preserve the coherence of the final response. Our two-stage process adapts RAG to this competitive and revenue-driven setting by augmenting dense retrieval with an LLM-based fine-ranking stage.

\textbf{LLM-based Recommendation.}
LLMs have shown strong context understanding for zero-shot ranking and recommendation~\citep{zhao2024recommender,li2024large}, often outperforming embedding-based methods in capturing nuanced user intent~\citep{liu2023chatgpt,dai2023uncovering}. However, these works focus solely on relevance or user satisfaction, without an economic layer. We integrate the LLM's relevance assessment as an organic compatibility score into a truthful auction, accommodating strategic bidding and incentive-compatible pricing.

\textbf{Externality in Ad Auctions.}
From classic position auctions~\citep{edelman2007internet,varian2007position} to deep learning-based mechanisms~\citep{zhang2021optimizing,liu2021neural} and whole-slate optimization~\citep{li2023learning,zhu2025contextual}, auction design has progressively addressed allocative externalities among ads. LLM-generated responses introduce a distinct generative externality: an inserted ad can alter the flow, tone, specificity, and length of the entire narrative. Our payment rule explicitly accounts for the dynamic organic probability of an ad within a generated response, extending externality-aware design to generative AI.

\section{Preliminary}\label{sec:pre}

We consider a practical scenario for single ad insertion. Given a user query $x$, there exists a candidate set of $N$ advertisers bidding for the opportunity to display their product within the response. 
Each advertiser $i$ has a \emph{public} description $d_i$ for its product and estimates a prior value-per-click $v_i$ based on the query, representing the expected value of a click should their ad be inserted. 
This value $v_i$ constitutes \emph{private} information; the platform observes only the reported bid $b_i$ and the public $d_i$ from each advertiser.

Formally, the input to the ad insertion mechanism comprises the query $x$, the set of ad descriptions $\{d_i\}_{i=1}^N$, and the reported bids $\{b_i\}_{i=1}^N$. The mechanism outputs a distribution over the allocation space (the response $y$) and a payment vector (the cost per click for each advertiser).

Let $\mathcal{Y}$ denote the space of possible responses. The mechanism is formally defined by a pair of functions $(g, p)$:
\begin{equation*}
    g: \mathcal{X} \times \mathcal{D}^N \times \mathbb{R}^N \to \Delta(\mathcal{Y}), \quad
    p: \mathcal{X} \times \mathcal{D}^N \times \mathbb{R}^N \to \mathbb{R}^N,
\end{equation*}
where $g$ is the allocation rule and $p$ is the payment rule. The final response is sampled from the distribution $g(x, \bm{d}, \bm{b})$. Upon the generation of response $y \sim g(x, \bm{d}, \bm{b})$, each advertiser forms a posterior value-per-click $\hat{v}_i(y)$, which may differ from the prior estimate $v_i$ derived solely from the query. The advertiser's utility $u_i$ is realized upon a user's click outcome $o_i \in \{0,1\}$, such that $u_i = (\hat{v}_i(y) - p_i) \cdot o_i$.

The allocation function $g$ inevitably relies on a foundation model for text generation. Although LLMs exhibit remarkable instruction-following capabilities~\cite{ouyang2022training}, expecting an off-the-shelf model to strictly adhere to scalar bids during allocation proves challenging, as LLMs often lack numerical sensitivity~\cite{lewkowycz2022solving, qian2023limitations}. While existing decoupled approaches compute feature scores using simple embedding similarity to guide the allocation~\cite{hajiaghayi2024ad}, such coarse features often fail to maximize user satisfaction and platform revenue. Therefore, we propose a two-stage score calculation method that leverages the LLM's own capabilities to better handle nuanced user intent, which we detail in the next section.

\subsection{Retrieval-Generation Decoupling}

A representative approach, proposed by Hajiaghayi et al.~\cite{hajiaghayi2024ad}, decomposes the allocation rule as $g = g_{\text{retrieval}} \odot g_{\text{generate}}$. For a single ad insertion, a winning advertiser $i^*$ is first determined by
\begin{equation*}
    i^* = \arg\max_{i} \big( s(x, d_i) \cdot b_i \big),
\end{equation*}
where the score $s(x, d_i)$ is computed via embedding similarity with optional noise:
\begin{equation*}
    s(x, d_i) = \text{Norm}(\text{Similarity}(\text{Emb}(x), \text{Emb}(d_i))) \cdot \text{Noise}_i,
\end{equation*}
and $\text{Norm}(\cdot)$ maps the raw similarity (e.g., cosine similarity in $[-1,1]$) to a strictly positive value, e.g., via min-max scaling to $(0,1]$, to avoid non‑positive bid multipliers. 
A predefined prompt template then incorporates the winning ad to generate the final response
$
    y \sim \text{GenAI}\big(\text{Prompt}(x, d_{i^*})\big).
$
The baseline second-price payment rule is then
\begin{equation*}
    p_i = \max_{j \neq i^*} \frac{s(x, d_j) \cdot b_j}{s(x, d_i)}.
\end{equation*}
Dubey et al.~\cite{dubey2024auctions} explore a related idea where the retrieval stage predicts a prominence vector over ads, but this requires the LLM to be sensitive to such a vector, which is difficult in practice.

This decoupled design requires only one expensive LLM call (the final generation), keeping the system lightweight. However, it places high demands on the accuracy of $s$: a poor score can make ad insertion counterproductive. To address this, we propose a two-stage scoring method that leverages the LLM's own judgment for more nuanced intent handling and diversity, while retaining the single-call full-text generation property. The next section details our mechanism with refined scoring and a critical-value payment rule.

\section{Two-Stage Scoring for Ad Auction}\label{sec:methods}

We first present the proposed mechanism for single ad insertion. 
The key idea is to retain the efficiency of retrieve-then-generate allocation while improving the relevance score used in the auction. 
To this end, the mechanism uses a two-stage ranking process, as is illudstrated in Figure~\ref{fig:framework}.
The first stage applies inexpensive embedding-based filtering to reduce the advertiser pool from $N$ to a small candidate set of size $K$, where $K \ll N$. 
The second stage uses the LLM's contextual judgment to refine the relevance scores only within this candidate set. 
The refined scores are then combined with advertisers' bids to determine the final allocation. 
We further design a critical-value payment rule for this two-stage selection process and discuss its extension to multiple ad insertions.

\subsection{Two-Stage Allocation for Single Ad Insertion}\label{sec:methods-allocation}

Given a user query $x$, advertiser descriptions $\{d_i\}_{i=1}^N$, and bids $\{b_i\}_{i=1}^N$, the first stage aims to efficiently identify a small set of plausible advertisers. 
Rather than directly embedding the raw query, we first prompt an LLM to generate a concise set of intent keywords, denoted by $\mathbf{w} \sim \text{GenAI}(\text{Prompt}_{\text{keywords}}(x))$. 
These keywords summarize the user's underlying intent and provide a cleaner signal for embedding-based retrieval.

For each advertiser $i$, we compute a Stage 1 relevance score using the normalized cosine similarity between the keyword embedding and the advertiser description embedding:
\begin{equation*}
    s_i^{(1)} = 
    \text{Norm}\big(
    \text{Similarity}(\text{Emb}(\mathbf{w}), \text{Emb}(d_i))
    \big).
\end{equation*}
The relevance score is then weighted by the advertiser's bid as $\text{Score}_i^{(1)} = s_i^{(1)} \cdot b_i$. 
The top-$K$ advertisers according to $\text{Score}_i^{(1)}$ form the candidate set:
\begin{equation*}
    \mathcal{S} = 
    \arg\underset{i \in \{1,\ldots,N\}}{\text{Top-}K} \; 
    \text{Score}_i^{(1)}.
\end{equation*}

The second stage refines the relevance assessment within $\mathcal{S}$ using the LLM itself. 
This stage is motivated by the fact that ad insertion often requires understanding implicit constraints, negative preferences, and intent shifts that may not be captured by embedding similarity alone. 
We present the original query $x$ and the candidate descriptions $\{d_j\}_{j \in \mathcal{S}}$ to the LLM and ask it to identify the most appropriate ad for insertion. 
To avoid forcing an irrelevant advertisement into the response, we include a no-insertion option $\phi$ and define the extended candidate set as $\hat{\mathcal{S}} = \mathcal{S} \cup \{\phi\}$.

Instead of relying on generated textual output, we extract the model's logits for the candidate labels. 
Specifically, the LLM produces a $(K+1)$-dimensional logit vector $\mathbf{z}$ from $\text{Prompt}_{\text{ranking}}(x,\{d_j\}_{j \in \hat{\mathcal{S}}})$. 
We normalize these logits with a softmax function to obtain the Stage 2 relevance score:
\begin{equation*}
    s_j^{(2)} = 
    \frac{\exp(z_j)}
    {\sum_{k \in \hat{\mathcal{S}}} \exp(z_k)}.
\end{equation*}
This score can be interpreted as the LLM's refined estimate of organic relevance within the candidate set. 
The final auction score is $\text{Score}_j^{(2)} = s_j^{(2)} \cdot b_j$, where we set $b_{\phi}=1$ for the no-insertion option. 
The winner is then selected as
\begin{equation*}
    i^* = \arg\max_{j \in \hat{\mathcal{S}}} \; \text{Score}_j^{(2)}.
\end{equation*}
If $i^*=\phi$, no advertisement is inserted. 
Otherwise, the winning advertiser's description is incorporated into the final generation prompt:
\begin{equation*}
    y \sim 
    \begin{cases}
    \text{GenAI}\big(\text{Prompt}(x, d_{i^*})\big), & \text{if } i^* \neq \phi,\\
    \text{GenAI}\big(\text{Prompt}(x)\big), & \text{if } i^* = \phi.
    \end{cases}
\end{equation*}

\subsection{Critical-Value Payment and Incentive Compatibility}\label{sec:methods-payment}

We next define the payment rule for the two-stage allocation process. 
Following critical-value pricing~\cite{riley1981optimal,klemperer1999auction}, the winning advertiser should pay the minimum bid required to remain the winner.

If no advertiser is selected, i.e., $i^*=\phi$, no payment is charged. 
Otherwise, let $i^*$ be the winning advertiser. 
Let $\text{Score}^{(1)}_{(K+1)}$ denote the $(K+1)$-th highest Stage 1 score, which is the threshold for entering the candidate set. 
Let $\text{Score}^{(2)}_{(2)}$ denote the second-highest Stage 2 score within $\hat{\mathcal{S}}$, which is the threshold for winning the final ranking. 
The payment of advertiser $i^*$ is defined as
\begin{equation*}
    p_{i^*} = \max \Bigg\{
    \underbrace{
    \frac{\text{Score}_{(K+1)}^{(1)}}{s^{(1)}_{i^*}}
    }_{\text{critical value for Stage 1}},
    \underbrace{
    \frac{\text{Score}_{(2)}^{(2)}}{s^{(2)}_{i^*}}
    }_{\text{critical value for Stage 2}}
    \Bigg\}.
\end{equation*}
Intuitively, the winner must pay enough to both enter the Stage 1 top-$K$ set and defeat all competitors in Stage 2. 
Taking the maximum of the two thresholds ensures that both requirements are satisfied.

We now state the incentive-compatibility property of the mechanism. 
As in Section~\ref{sec:pre}, advertiser $i$ has a prior value-per-click estimate $v_i$, while the realized value $\hat{v}_i(y)$ depends on the generated response. 
Let $o_i(y)$ denote the click outcome or click probability associated with response $y$. 
The expected utility of the winning advertiser is
$
    u_i = 
    \mathbb{E}_{y \sim \text{GenAI}(\text{Prompt}(x, d_i))}
    \big[
    (\hat{v}_i(y)-p_i)\cdot o_i(y)
    \big].
$

\begin{restatable}{proposition}{truthfulProp}\label{thm:truthful}
    Assume that the relevance scores used by the allocation rule are independent of advertiser $i$'s reported bid except through the bid-weighted ranking rule. 
    Further assume that $o_i(y)$ and $\hat{v}_i(y)$ are independent conditioned on the ad being displayed, and that the prior estimate is unbiased, i.e., $v_i=\mathbb{E}_y[\hat{v}_i(y)]$. 
    Then the two-stage auction mechanism with the payment rule above is incentive compatible for a utility-maximizing advertiser.
\end{restatable}

\begin{proof}[Proof sketch]
    Under the independence and unbiasedness assumptions, advertiser $i$'s expected utility can be written as
    \begin{equation*}
        u_i = 
        \mathbb{E}_y[(\hat{v}_i(y)-p_i)\cdot o_i(y)]
        =
        (v_i-p_i)\cdot \bar{o}_i,
    \end{equation*}
    where $\bar{o}_i=\mathbb{E}_y[o_i(y)]$ is the expected click probability. 
    Since $\bar{o}_i$ is nonnegative, the advertiser's bidding decision reduces to a standard critical-value comparison between value and payment. 
    By construction, $p_{i}^{\mathrm{crit}}$ is the minimum bid required to pass both the Stage 1 entry threshold and the Stage 2 winning threshold. 
    Bidding above this value yields the same allocation and payment, while bidding below it loses the allocation. 
    Therefore, truthful bidding weakly dominates both underbidding and overbidding. 
    The full proof is deferred to Appendix~\ref{sec:appendix:truthful-proof}.
\end{proof}

This payment rule differs from a classic single-stage second-price auction because the final-stage competitors are themselves selected by the first-stage auction. 
An advertiser must therefore satisfy two constraints: being included in the candidate set and winning the LLM-based ranking. 
The proposed payment captures both constraints through the maximum of the two critical values.

\subsection{Extension to Multiple Ad Insertions}\label{sec:methods-multi}

The same two-stage mechanism can be extended to multiple ad insertions, either across turns in a dialogue or within a long-form response. 
Following the retrieve-then-generate workflow, we divide the response into $T$ segments. 
For segment $t$, the auction is conditioned on the original query $x$ and the response prefix $y_{<t}$. 
The mechanism selects either an advertiser $i_t^*$ or the no-insertion option, after which the LLM generates the next segment from $\text{Prompt}(x,y_{<t},d_{i_t^*})$, with $d_{i_t^*}$ omitted when $i_t^*=\phi$. 
The final response is the concatenation $y = y_1 \oplus y_2 \oplus \cdots \oplus y_T$.

The main additional cost comes from three components: reprocessing the prefix for each segment, keyword generation for Stage 1, and logit computation for Stage 2. 
In practice, $T$ is small because each segment typically contains several sentences. 
Prefix reprocessing is comparable to a small number of additional decoding steps, while Stage 2 requires only a single forward pass over a small candidate set. 
Keyword generation is the most visible extra cost, but it produces only a short sequence and can be traded off against retrieval quality. 
Thus, the overall latency overhead remains controllable, as empirically verified in Section~\ref{sec:experiment-diversity}.

\section{Experimental Results}\label{sec:experiment}

In this section, we evaluate \name\ against an Embedding-only baseline and several ablation variants. Across different model scales, we show that \name\ improves single ad selection accuracy. In more realistic long-form responses with multiple ad insertions, \name\ further achieves a favorable balance between user satisfaction and inference efficiency.

\subsection{Experimental Setup}\label{sec:experiment-setup}

\paragraph{Datasets.}
All experiments are conducted on a synthetic advertiser dataset consisting of 100 advertisers from 8 product categories. The advertisers are curated by prompting an LLM. Each advertiser is associated with a brand description of approximately 20 words, which serves as the reference text for both embedding-based retrieval and LLM-based scoring. Representative examples are provided in Appendix~\ref{sec:appendix:details_of_implementation}.

\paragraph{Methods.} 
We compare methods under the retrieve-then-generate paradigm. The retrieval and ranking strategies are as follows. \textbf{Embedding-only}~\citep{hajiaghayi2024ad} ranks advertisers solely by the embedding similarity between the user query and advertiser descriptions. \textbf{\name} is our proposed two-stage method, which first performs coarse filtering with embedding similarity and then applies LLM logits for fine-grained ranking. \textbf{\name\ (LLM-only)} removes the coarse filtering stage and directly applies LLM scoring to all $N$ advertisers. \textbf{\name\ (Keyword-Match)} uses LLM-generated keywords for embedding-based retrieval but omits the fine-grained LLM-logit scoring stage.

\paragraph{Implementation.}
All experiments use uniform bids ($b_i = 1.0$) to isolate the allocation and scoring mechanisms from auction dynamics under heterogeneous bidding.
We distinguish between two types of models used in our system: the \textit{Embed Model}, which produces embeddings for retrieval, and the \textit{Gen Model}, which is used for keyword generation, logit-based scoring, and final response generation. For the Gen Model, we evaluate three model scales: 4B, 30B, and 235B parameters. All experiments are conducted on NVIDIA A800 GPUs. Since our experiments only require LLM inference, we deploy all LLM inference services with SGLang~\citep{zheng2024sglang}, which better reflects the behavior of modern high-throughput inference engines.

\begin{table*}[t]
    \centering
    \caption{\textbf{Accuracy of single-ad insertion across varying strategies and model scales.} We evaluate the selection accuracy of Embed Models and Gen Models on three distinct simulated scenarios: Complex Intent Analysis, Natural Intent Shift, and Negative Product Refinement. \name\ consistently outperforms the embedding-only baseline.}
    \resizebox{0.94\textwidth}{!}{
    \begin{tabular}{lcc|cccc}
        \toprule
        \textbf{Method} & \textbf{Embed Model} & \textbf{Gen Model} & \textbf{Complex} & \textbf{Shift} & \textbf{Negative} & \textbf{Overall}\\
        \midrule
        \multirow{2}{*}{Embedding-only} & all-MiniLM-L6-v2 & N/A & 67.5\% & 68.8\% & 52.5\% & 62.93\% \\
        & Qwen3-Embedding-8B & N/A & 55.0\% & 60.0\% & 53.8\% & 56.3\% \\
        \midrule
        \multirow{3}{*}{LERA (LLM-only)} & N/A & Qwen3-4B & 76.2\% & 71.2\% & 72.5\% & 73.3\% \\
         & N/A & Qwen3-30B-A3B & 97.5\% & 86.2\% & 100\% & 94.57\% \\
         & N/A & Qwen3-235B-A22B & 100\% & 100\% & 100\% & 100\% \\
        \midrule
        \multirow{3}{*}{LERA (Keyword-Match)} & all-MiniLM-L6-v2& Qwen3-4B & 67.5\% & 73.8\% & 82.5\% & 74.6\%\\
         & all-MiniLM-L6-v2 & Qwen3-30B-A3B & 81.2\% & 90.0\% & 88.8\% & 86.67\% \\
         & all-MiniLM-L6-v2 & Qwen3-235B-A22B & 75.0\% & 81.2\% & 86.2\% & 80.8\% \\
        \midrule
        \multirow{3}{*}{LERA} & all-MiniLM-L6-v2 & Qwen3-4B & 91.2\% & 97.5\% & 95.0\% & 94.6\% \\
         & all-MiniLM-L6-v2 & Qwen3-30B-A3B & 96.2\% & 100\% & 100\% & 98.73\% \\
             & all-MiniLM-L6-v2 & Qwen3-235B-A22B & 96.2\% & 98.8\% & 97.5\% & 97.5\% \\
            \bottomrule
        \end{tabular}
    }
    \label{tab:single_ad_insertion}    
\end{table*}

\subsection{Selection Accuracy in Single Ad Insertion}\label{sec:experiment-accuracy}

We first evaluate single ad selection accuracy, which serves as the elementary allocation task in LLM-based chatbot advertising. Given a user query, each strategy is required to select the advertiser that is most suitable for insertion into the generated response.

To comprehensively evaluate different strategies, we construct three types of test queries that reflect different user intent scenarios:
\begin{enumerate}
\item \textbf{Complex Intent Analysis} (80 queries): queries containing nuanced or implicit constraints that require multifaceted reasoning.
\item \textbf{Natural Intent Shift} (80 queries): queries that initially mention one product or category but naturally shift toward a different target need.
\item \textbf{Negative Product Refinement} (80 queries): queries that explicitly reject an initially mentioned product before stating the actual need.
\end{enumerate}
Although ad insertion can be subjective in general, each test query in our benchmark is constructed to have a unique ground-truth advertiser in the advertiser pool. This allows us to report selection accuracy as a direct metric. Representative examples are provided in Appendix~\ref{sec:appendix:dataset}.

Table~\ref{tab:single_ad_insertion} reports the selection accuracy across retrieval strategies and model scales. The Embedding-only baseline achieves relatively low matching accuracy, even when using an 8B embedding model. This is expected: embedding models are primarily trained to capture semantic proximity, whereas ad selection often requires reasoning over intent shifts, explicit exclusions, and implicit constraints. In contrast, \name\ consistently improves selection accuracy across settings.

The advantage of \name\ is particularly pronounced for smaller generative models, such as the 4B model. In this regime, the two-stage design provides stronger guidance than either pure embedding retrieval or keyword-only retrieval. Moreover, for the 4B and 30B models, directly ranking all 100 candidate advertisers with LLM logits, i.e., \name\ (LLM-only), performs worse than our two-stage method. This suggests that the coarse filtering stage does not merely improve efficiency; it also reduces the difficulty of the fine-grained ranking problem by presenting the LLM with a smaller and more relevant candidate set.

On the other hand, \name\ (LLM-only) benefits most clearly from increased model capability. With the 235B model, this variant achieves perfect accuracy across all test categories. However, it is poorly suited to practical high-throughput inference engines such as SGLang and vLLM. These systems do not directly expose full logits over large candidate sets in the most efficient generation path, and computing logits over all advertisers requires parallel prefilling across the candidate set. This creates a substantial throughput bottleneck in real-time applications. We provide a detailed discussion in Appendix~\ref{sec:appendix:sglang_implementation}, and the latency results in the next subsection further demonstrate the overhead of \name\ (LLM-only).

For \name, we observe that the 235B model is slightly worse than the 30B model in some settings. In our experiments, this is mainly because the larger model sometimes adopts a more conservative behavior and chooses not to insert an advertisement when the candidates appear marginally relevant. Under our accuracy metric, such outcomes are counted as incorrect, which explains why the 235B model can underperform the 30B model in these cases.

\subsection{Multi Ad Insertion in Dynamic Chatting Flow}\label{sec:experiment-diversity}

We next evaluate multiple ad insertion, a more realistic setting in which ad allocation must be performed dynamically throughout a generated response. In real chatbot interactions, ad selection should evolve with the generated context. An ideal strategy should decide not only which ad to insert, but also whether an additional ad is still needed after previous insertions. It should also avoid repeatedly selecting advertisers from the same category when the user requests recommendations for multiple distinct needs.

To evaluate this setting, we construct 50 queries in which users explicitly request recommendations from three distinct ad categories. Following the multi-ad insertion workflow of Hajiaghayi et al.~\cite{hajiaghayi2024ad}, we split response generation into multiple segments. Before each segment, an auction is conducted to decide whether to insert an ad and which advertiser to select. The model then continues the response for two or three sentences using a prompt conditioned on the auction outcome. This process is repeated until the full response is generated. For a fair comparison, we limit the maximum number of inserted ads to three and set the number of first-stage winners in \name\ to $K=8$ for the multi-ad dynamic setting (we use $K=5$ for the single-ad setting reported in Table~\ref{tab:single_ad_insertion}).

\begin{wraptable}{r}{0.55\textwidth}
    \centering
    \caption{\textbf{Query satisfaction in dynamic multi-ad insertion scenarios.} 
    The metric represents the average proportion of user-requested categories successfully fulfilled by the inserted ads within a limit of three insertion opportunities. 
    \name\ achieves the best balance of context-aware diversity compared to ablation variants.}
    \resizebox{0.58\textwidth}{!}{
    \begin{tabular}{l|ccc}
        \toprule
        \textbf{Method} & 4B & 30B-A3B & 235B-A22B \\
        \midrule
        Embedding-only & 46.67\% & 42.00\% & 45.33\% \\
        LERA (LLM-only) & 48.61\% & 64.00\% & 86.67\% \\
        LERA (Keyword-Match) & 64.67\% & 63.33\% & 64.67\% \\
        LERA (Cosine-LLM) & 63.33\% & 58.67\% & 58.00\% \\
        LERA & 72.00\% & 80.00\% & 80.67\% \\
            \bottomrule
        \end{tabular}
    }
    \label{tab:multi_ad_insertion}    
\end{wraptable}

\paragraph{Query Satisfaction.}
We first measure how well the inserted ads satisfy user intent by computing the average proportion of requested categories that are fulfilled. Although the generated response may verbally address all user needs, this metric focuses on whether monetized ad insertions cover the user's requested categories. It therefore reflects the commercial effectiveness of the allocation strategy.

As shown in Table~\ref{tab:multi_ad_insertion}, \name\ improves over the Embedding-only baseline by 20\% to 40\% across the 4B to 235B models. This demonstrates that the proposed two-stage auction remains effective in a dynamic multiple insertion setting. Among the ablation variants, \name\ (Keyword-Match) also improves over the baseline, suggesting that LLM-generated keywords help bridge the gap between raw user queries and advertiser descriptions. However, keyword-based embedding retrieval alone remains insufficient, indicating the importance of the second-stage LLM-logit ranking.

We further include \name\ (Cosine-LLM), an additional ablation that uses cosine similarity rather than keyword similarity for the initial filtering stage. This variant also improves over the Embedding-only baseline, but still lags substantially behind \name. Since we set $K=5$, a coarse filtering stage based only on standard embedding similarity may exclude the preferred advertiser before LLM scoring, especially when the advertiser pool contains many semantically similar candidates. Together, the performance of these ablations supports the design of \name: effective ad allocation requires both a query-aware coarse filtering stage and a fine-grained LLM-based ranking stage.

We also report \name\ (LLM-only) as a special reference method. This method relies entirely on the model's own reasoning ability, and its performance therefore improves markedly as the Gen Model becomes stronger. With the 235B model, it achieves the best query satisfaction. However, as discussed below, this accuracy-oriented advantage comes with significant inference overhead.

\begin{figure*}[t]
    \centering
    \begin{minipage}[t]{0.48\textwidth}
        \centering
        \includegraphics[width=\linewidth]{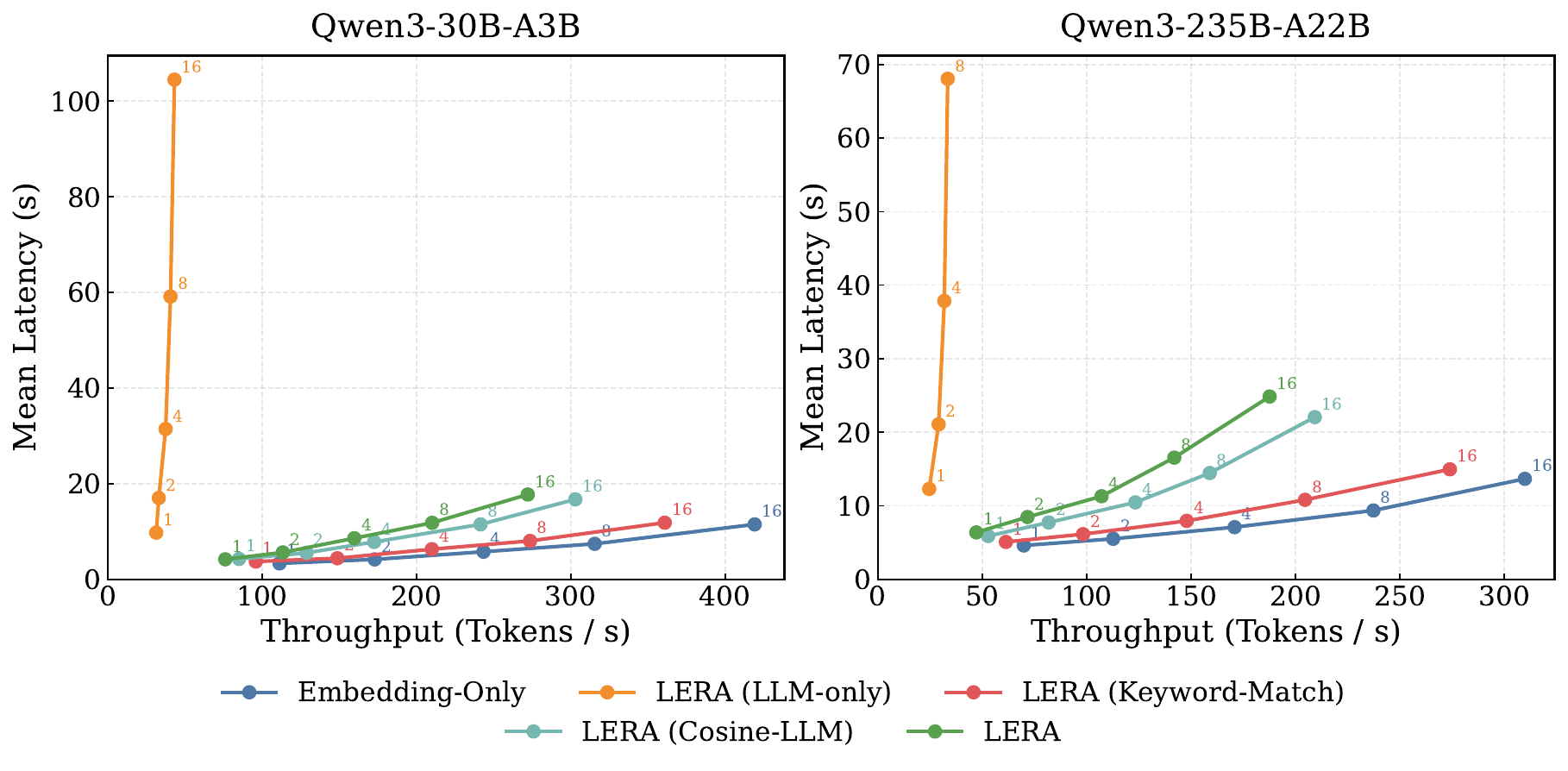}
        \captionof{figure}{\textbf{Latency and throughput tradeoff.} Comparison of \name\ and its ablation variants under different concurrency levels using 30B and 235B Gen Models. Each point corresponds to a concurrency configuration and illustrates the efficiency cost of more expensive auction mechanisms.}
        \label{fig:throughput}
    \end{minipage}
    \hfill
    \begin{minipage}[t]{0.48\textwidth}
        \centering
        \includegraphics[width=\linewidth]{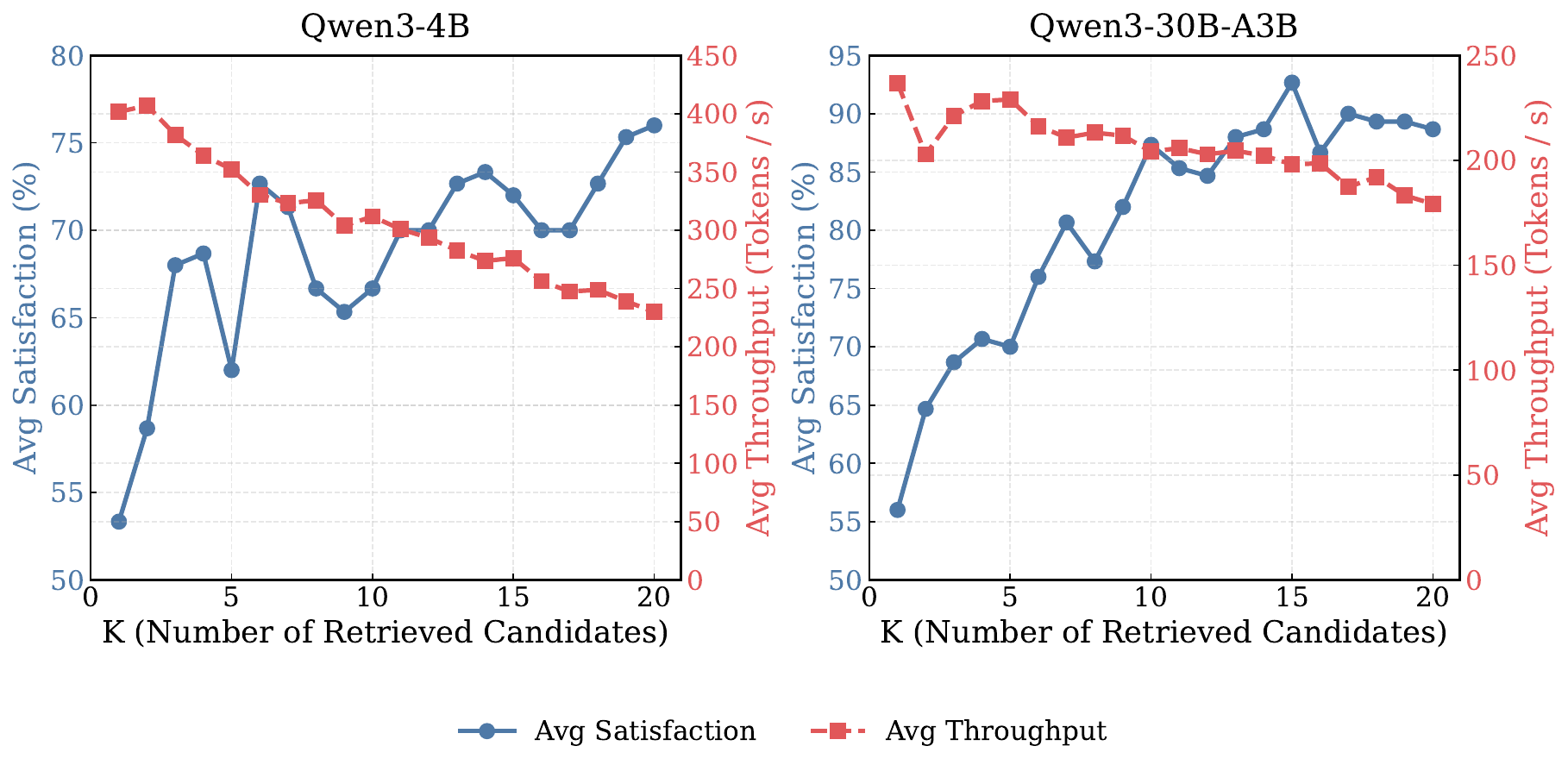}
        \captionof{figure}{\textbf{Effect of the first-stage winner set size $K$.} Ablation of the top-$K$ coarse filtering budget for LERA with 4B and 30B Gen Models. Increasing $K$ generally improves allocation quality but reduces throughput, revealing the tradeoff between candidate coverage and inference efficiency.}
        \label{fig:topk_ablation}
    \end{minipage}
\end{figure*}

\paragraph{Throughput and Latency.}
To evaluate deployability and practical efficiency, we measure latency, the overall time required to generate a response including the auction process, and throughput, the average number of generated tokens per second.
% These two metrics naturally trade off under different concurrency levels: higher concurrency can improve aggregate throughput, but it also increases the latency experienced by each query. 
Following standard inference system evaluations, we compare latency and throughput curves across models and methods by varying the concurrency level from 1 to 16.

As shown in Figure~\ref{fig:throughput}, \name\ (LLM-only) quickly reaches a throughput bottleneck because SGLang computes logits through parallel prefilling over a large candidate set. 
Although this variant performs well in accuracy-oriented evaluations, its computational cost makes it difficult to deploy in real-time auction systems. 
For \name, when the concurrency level is low, e.g., below 4, there is moderate throughput reduction and latency overhead compared with the Embedding-only baseline. 
At higher concurrency levels, however, the latency overhead becomes more visible. 
% For example, with the 235B model and a latency budget of 15 seconds, \name\ achieves around 130 tokens/s, whereas the baseline maintains approximately 300 tokens/s. 
This highlights a practical tradeoff: \name\ substantially improves ad matching quality, but it also consumes additional inference resources. 
In deployment scenarios, this cost can be justified when improved matching accuracy are transformed into higher platform revenue or better advertiser utility.
Designing auction mechanisms that better align with modern inference engines remains an important direction for future work.

\paragraph{Ablation Study on the First-Stage Candidate Set Size $K$.}
Finally, we study how the first-stage candidate set size $K$ affects \name. We fix the concurrency level to 8 and vary $K$ to evaluate its impact on both query satisfaction and throughput.
As shown in Figure~\ref{fig:topk_ablation}, increasing $K$ generally improves performance because the second-stage LLM scorer has access to a broader candidate set. However, once $K$ becomes sufficiently large, the performance gain diminishes and may even slightly decrease. One possible reason is that a larger candidate set introduces less relevant advertisers into the second-stage ranking problem, making fine-grained LLM scoring more difficult. In terms of efficiency, throughput decreases as $K$ increases because more candidates require LLM-based scoring. These results suggest that $K$ should be chosen to balance allocation quality and inference efficiency, rather than being set as large as possible.

% \vspace{-10pt}
\section{Conclusion}\label{sec:conclusion}

In this work, we presented a novel two-stage auction mechanism for advertisement insertion in LLM-based chatbots that effectively balances semantic relevance with commercial objectives. Our approach combines the efficiency of embedding-based coarse filtering with the semantic sophistication of LLM-based fine-grained ranking, addressing key limitations of existing methods. The proposed framework demonstrates superior performance in handling complex user intents, maintaining ad diversity, and preserving reasonable computational overhead. Future work could explore adaptive mechanisms for dynamically adjusting the balance between semantic relevance and commercial value based on contextual factors.

\section*{Limitations}
This paper studies auction design within the retrieve-then-generate paradigm. Due to the lack of an off-the-shelf open-source dataset, our experiments rely on a synthetic dataset constructed by an LLM. The theoretical results are based on an independence assumption, which we discuss further in Appendix~\ref{sec:appendix:truthful-proof}. Despite these limitations, we believe the proposed framework represents a step toward the practical deployment of auction-based ad insertion in LLM-powered chatbots.

\section*{Broader Impacts}
As \name\ has strong practical significance, we believe it can be used to enhance the commercial effectiveness of LLM-based chatbots. However, we also acknowledge that the proposed framework may have potential negative societal impacts. For example, it could be used to insert advertisements into user responses, which might be perceived as intrusive or annoying. Therefore, careful design of the auction mechanism is necessary to mitigate such negative societal impacts when the method is deployed in real-world applications.

\bibliographystyle{plain}
\bibliography{references}
\appendix

\appendix

\section{Appendix Overview}

This appendix provides additional materials supporting the main paper. Appendix~\ref{sec:related} expands the discussion of related work. Appendix~\ref{sec:appendix:details_of_implementation} reports implementation details and generation settings for reproducibility. Appendix~\ref{sec:appendix:dataset} describes the construction of the synthetic advertiser-query benchmark. Appendix~\ref{sec:appendix:prompts} lists the prompt templates used in different stages of the framework. Appendix~\ref{sec:appendix:examples} provides representative case studies, and Appendix~\ref{sec:appendix:truthful-proof} contains the omitted proof of Proposition~\ref{thm:truthful}.

\section{Additional Related Work}\label{sec:related}

This section provides a more detailed discussion of the literature related to our work. Our framework lies at the intersection of LLM-mediated advertising, retrieval-augmented generation, LLM-based recommendation, and externality-aware auction design. We therefore organize the discussion around four themes: ad auctions in LLM interfaces, retrieval-augmented generation, ad recommendation with LLMs, and externalities in auction mechanisms.

\subsection{Ad Auctions in LLMs}

Research on integrating advertising auctions into LLM-based chatbots is a rapidly emerging area. Existing frameworks can generally be categorized into \textit{generate-then-retrieve} and \textit{retrieve-then-generate} paradigms, depending on whether the auction is applied after a candidate response is generated or before the final response is produced.

In the \textit{generate-then-retrieve}, or post-generation, approach, the model first generates candidate responses or probability distributions, which are then used as the basis for auction decisions. Duetting et al.~\cite{duetting2024mechanism} pioneered a mechanism in which advertisers bid on token probabilities, aggregating preferences from multiple LLMs to produce a joint output. Similarly, Mordo et al.~\cite{mordo2024sponsored} and Soumalias et al.~\cite{soumalias2024truthful} explore auctioning the ``right to speak'' or inserting ads into generated placeholders. While these methods are theoretically grounded, they often require multiple inference passes or complex output aggregation, leading to high latency and inference costs~\citep{tang2024genai}.

In contrast, the \textit{retrieve-then-generate} paradigm, recently formalized by Hajiaghayi et al.~\cite{hajiaghayi2024ad}, treats advertisers as retrievable documents within a RAG framework. An auction determines which ads are retrieved and inserted into the prompt context before generation begins. This approach is modular and computationally efficient because it decouples the auction stage from the expensive generation process. However, existing implementations often rely heavily on standard retrieval scores, which may not fully align with commercial objectives. This limitation motivates our framework, which introduces a selection mechanism that explicitly combines coarse retrieval with LLM-based ad relevance estimation.

More recently, Zhao et al.~\cite{zhao2025llm} proposed a training method that enables an LLM to generate responses satisfying bid monotonicity. Although this approach combines retrieval and generation, it relies on costly model training and may be less flexible when the ad inventory changes frequently. Our work instead focuses on a modular inference-time mechanism that can be applied without retraining the underlying LLM.

\subsection{Retrieval-Augmented Generation}

Retrieval-augmented generation was introduced to ground LLM responses in external and up-to-date knowledge~\citep{lewis2020retrieval}. A standard RAG pipeline retrieves documents based on vector similarity, often using dense embeddings, and then feeds the retrieved documents into the generator~\citep{gao2023retrieval}. This paradigm has been widely adapted for downstream tasks, including recommendation systems where items are retrieved as contextual evidence~\citep{csakar2025maximizing}.

As discussed in the main paper, RAG offers a natural abstraction for ad insertion: advertisements can be viewed as external knowledge chunks that may be inserted into the LLM context. However, a fundamental difference remains between ordinary RAG and sponsored generation. Standard RAG primarily aims to maximize semantic fidelity to the user query, whereas ad insertion must also account for revenue, advertiser value, and the likelihood that an ad can be integrated without degrading user experience. In this sense, the relevant retrieval signal is not merely semantic similarity, but a broader notion of commercial compatibility. Our two-stage retrieval and scoring process adapts standard RAG principles to this competitive and revenue-driven setting.

\subsection{Ad Recommendation with LLMs}

Beyond auction design, a growing body of work studies LLMs as ranking and recommendation engines. Recent surveys suggest that LLMs possess strong semantic understanding that can be leveraged for zero-shot item recommendation and ranking~\citep{zhao2024recommender, li2024large}. Liu et al.~\cite{liu2023chatgpt} and Dai et al.~\cite{dai2023uncovering} show that LLMs can evaluate the relevance of candidate items to a user's query and history, often outperforming traditional embedding-based methods in capturing nuanced user intent.

The key distinction between this literature and our work is the economic layer. LLM-based recommendation studies typically focus on maximizing relevance, utility, or user satisfaction, while abstracting away from strategic advertiser behavior and incentive-compatible pricing. In contrast, our framework uses the LLM's relevance assessment as an organic compatibility score within a truthful auction mechanism. This allows us to preserve the semantic strengths of LLM-based recommendation while incorporating advertiser bids and payment rules.

\subsection{Externality in Ad Auction Design}

Early work on sponsored search auctions, such as the Generalized Second Price auction, primarily modeled externalities as position effects, assuming that an ad's click probability depends on its slot and inherent quality~\citep{edelman2007internet, varian2007position}. As advertising shifted from search pages to recommendation feeds and mixed-content environments, later work began to study richer \textit{allocative externalities}, where an ad's value and user impact depend on neighboring content.

To address such effects, researchers have proposed deep learning-based mechanisms. \textit{Deep GSP}~\citep{zhang2021optimizing} and \textit{Neural Auctions}~\citep{liu2021neural} use deep networks to learn non-linear ranking scores that capture inter-item dependencies while maintaining incentive compatibility through monotone allocation rules. More recent mechanisms, including WCVCG~\cite{li2023learning}, \textit{MIAA}~\citep{li2024deep}, and \textit{Contextual Generative Auctions}~\citep{zhu2025contextual}, explicitly model whole-slate effects by jointly optimizing the arrangement of ads and organic items to balance platform revenue with user experience.

While these works successfully handle externalities in static lists or feeds, LLM responses introduce a different type of \textit{generative externality}. Unlike a fixed feed, an LLM response is a cohesive narrative in which an ad insertion can alter the text's flow, tone, specificity, and length~\citep{balseiro2025position}. Our work extends externality-aware design to this generative setting by proposing a payment rule that accounts for the dynamic organic probability of an ad within a generated response.

Overall, existing work has studied auction design for LLM-generated content, LLM-based recommendation, and externalities in traditional ad allocation settings. However, these lines of work have not fully addressed the setting where advertisements are selected before or during LLM generation and where the final response itself determines the organic compatibility of an ad. Our work complements this literature by introducing a retrieval-based auction framework that combines efficient candidate selection, LLM-based relevance estimation, and an externality-aware payment rule tailored to generative responses.

\section{Implementation Details}\label{sec:appendix:details_of_implementation}

We provide implementation details to facilitate reproducibility and to clarify how the proposed mechanism is instantiated in practice. Unless otherwise stated, all methods are evaluated under the same model, retrieval, and generation settings, so that performance differences can be attributed primarily to the auction and selection strategies rather than to implementation-specific choices.

Our experiments are implemented using the SGLang framework. We use SGLang primarily to leverage its RadixAttention mechanism for efficient KV-cache reuse, which is important for our multi-turn dynamic insertion strategies. In these strategies, the system repeatedly evaluates whether a new advertisement should be inserted as the response unfolds, and therefore benefits from reusing previously computed context representations.

\begin{table}[t]
\centering
\small
\caption{Summary of implementation settings.}
\label{tab:implementation_settings}
\begin{tabularx}{0.95\linewidth}{lY}
\toprule
\textbf{Component} & \textbf{Setting} \\
\midrule
LLM backbone & Qwen models, including Qwen-4B, Qwen-30B, and Qwen-235B, are used for auction decision-making and response generation. \\
Embedding model & \texttt{sentence-transformers/all-MiniLM-L6-v2} is used for cosine-similarity filtering and combined retrieval strategies. \\
Serving framework & SGLang is used to leverage RadixAttention for efficient KV-cache reuse in dynamic multi-turn insertion. \\
Response generation & Temperature is set to $0.7$; the maximum generation length is $512$ tokens for single-turn generation and $64$ tokens per chunk for dynamic generation. \\
Auction scoring & Temperature is set to $0.0$ during LLM-based auction scoring to obtain deterministic log-probabilities for Q-score computation. \\
Dynamic insertion & The chat-based dynamic strategy is limited to at most $10$ generation chunks and at most $3$ ad insertions per session. \\
Initial retrieval & The top-$K$ retrieval parameter is set to $K=5$ in the single-turn setting and $K=8$ in the dynamic setting. \\
\bottomrule
\end{tabularx}
\end{table}

These settings are designed to separate the effects of retrieval and auction design from those of stochastic decoding. In particular, deterministic scoring is used for the auction stage, while moderate-temperature generation is retained for producing natural conversational responses.

\subsection{Detailed Inference Hyperparameters}
Table~\ref{tab:inference_hyperparameters} provides a comprehensive list of inference hyperparameters used in our experiments, extracted from the implementation code for full reproducibility.

\begin{table}[h]
\centering
\small
\caption{Detailed inference hyperparameters.}
\label{tab:inference_hyperparameters}
\begin{tabularx}{0.95\linewidth}{lX}
\toprule
\textbf{Hyperparameter} & \textbf{Value / Setting} \\
\midrule
\multicolumn{2}{l}{\textbf{General Settings}} \\
Device & \texttt{cuda} (NVIDIA A800 GPUs) \\
Random seed & $42$ \\
Maximum concurrent queries & $8$ \\
\midrule
\multicolumn{2}{l}{\textbf{Response Generation}} \\
Generation temperature & $0.7$ \\
Max new tokens (single-turn) & $512$ \\
Max new tokens (dynamic per chunk) & $64$ \\
Sampling strategy & Top-p sampling (default SGLang) \\
\midrule
\multicolumn{2}{l}{\textbf{Keyword Generation (Stage 1)}} \\
Keyword count range & $2$--$3$ keywords \\
Keyword gen temperature & $0.7$ \\
Keyword gen max tokens & $64$ \\
\midrule
\multicolumn{2}{l}{\textbf{LLM Auction Scoring (Stage 2)}} \\
Scoring temperature & $0.0$ (deterministic) \\
Logprob source & SGLang \texttt{normalized\_prompt\_logprobs} \\
Choice symbols & Numeric (``1'', ``2'', ...) or single-token when enabled \\
No-insertion symbol & ``0'' \\
No-insertion margin & $0.0$ \\
Score normalization & Softmax over logprobs \\
\midrule
\multicolumn{2}{l}{\textbf{Two-Stage Filtering}} \\
Top-$K$ candidates (single-turn) & $K=5$ \\
Top-$K$ candidates (dynamic) & $K=8$ \\
\midrule
\multicolumn{2}{l}{\textbf{Chat-Based Dynamic Strategy}} \\
Maximum generation chunks & $10$ (\texttt{max\_parts}) \\
Maximum ad insertions & $3$ (\texttt{max\_insertions}) \\
System prompt & ``You are a helpful assistant...'' (see Appendix~\ref{sec:appendix:prompts}) \\
Response termination & Empty, whitespace, or ``END'' token \\
\midrule
\multicolumn{2}{l}{\textbf{Bid Sampling}} \\
Alternative & Uniform $1.0$ (\texttt{--all\_one\_bids} flag) \\
\bottomrule
\end{tabularx}
\end{table}

\subsection{Logit Computation with SGLang}\label{sec:appendix:sglang_implementation}

As introduced in Section~\ref{sec:experiment-setup}, we utilize SGLang~\cite{zheng2024sglang} for our LLM inference to reflect the performance characteristics of modern, high-throughput deployment environments. However, extracting log probabilities (log probs) for the second-stage scoring requires specific engineering considerations due to the design of such inference engines.

In the second stage, the mechanism evaluates $K+1$ candidate choices to compute refined relevance scores. Standard inference libraries, such as Hugging Face \texttt{transformers}, allow direct access to the model's raw computational results, enabling the extraction of full vocabulary logits in a single forward pass. In contrast, highly optimized serving frameworks like SGLang prioritize overall throughput and memory management, and therefore do not natively expose intermediate logits for arbitrary token choices during a standard generation call.

To obtain the required log probabilities, SGLang evaluates multiple choices through parallel prefilling. Under this mechanism, $K+1$ distinct prompt variations are constructed—each explicitly appending one of the candidate selections. The inference engine processes these prompts concurrently and returns the sequence-level log probabilities, which are then used directly to derive the final Stage 2 auction scores. 

Although the Hugging Face engine could theoretically achieve this logit extraction with fewer forward passes, its inference efficiency is significantly lower. Because most practical, large-scale LLM services today are deployed on high-throughput engines like SGLang or vLLM, choosing SGLang for our experiments ensures that the measured efficiency metrics, latency tradeoffs, and observed throughput bottlenecks accurately reflect real-world production constraints.

\section{Dataset Construction Details}\label{sec:appendix:dataset}

This section describes the construction of our synthetic benchmark for evaluating ad selection in LLM-mediated responses. The goal of the dataset is not to reproduce a specific commercial ad inventory, but to create controlled query-ad matching scenarios in which semantic relevance, commercial intent, negation, and intent shifts can be systematically evaluated.

The dataset contains advertiser descriptions and user queries, both generated with the assistance of LLMs and manually organized into controlled evaluation groups. Advertiser descriptions are designed to be short enough to fit naturally into prompts while still containing specific product attributes that allow fine-grained matching. User queries are constructed to test whether a method can distinguish true commercial intent from superficial lexical or semantic overlap.

\subsection{Advertiser Categories and Descriptions}

We define eight product categories and generate multiple advertiser descriptions for each category, resulting in 100 unique advertisers. Each advertiser description is approximately 20--30 words long and highlights product features, target users, or use cases. Table~\ref{tab:dataset_categories} lists the categories and a representative advertiser description for each.

\begin{table}[t]
    \centering
    \small
    \caption{Advertiser categories and representative advertiser descriptions.}
    \label{tab:dataset_categories}
    \begin{tabularx}{0.95\linewidth}{lY}
        \toprule
        \textbf{Category} & \textbf{Representative Description} \\
        \midrule
        Snacks & \textbf{CrispFalls Potato Chips.} CrispFalls specializes in classic, ridged potato chips with a crisp snap, targeting traditional snackers who prefer salty crunch and familiar flavors. \\
        Beverages & \textbf{VelvetBean Cold Brew.} VelvetBean offers nitrogen-infused cold brew coffee sourced from single-origin Ethiopian beans, designed for coffee drinkers seeking a naturally sweet and creamy texture. \\
        Electronics & \textbf{SilentStorm Noise-Cancelling Earbuds.} SilentStorm produces true wireless earbuds with active noise cancellation for commuters and focus-seekers who want immersive audio. \\
        Tops & \textbf{SummitShell Rain Jacket.} SummitShell designs lightweight, waterproof, and breathable rain jackets for outdoor adventurers and urban commuters. \\
        Footwear & \textbf{TerraTread Hiking Boots.} TerraTread manufactures rugged waterproof hiking boots with deep-lug soles for backpackers and hikers on technical terrain. \\
        Skincare & \textbf{CitraBlast Vitamin C Serum.} CitraBlast formulates stabilized Vitamin C serum for brightening and antioxidant protection, targeting users concerned with skin dullness. \\
        Home \& Living & \textbf{NimbusLight Desk Lamp.} NimbusLight designs adjustable smart desk lamps with color-temperature control for students and professionals needing glare-free lighting. \\
        Food \& Dining & \textbf{OceanPrime Seared Salmon.} OceanPrime offers ready-to-eat, sustainably sourced Atlantic salmon fillets prepared using a sous-vide technique. \\
        \bottomrule
    \end{tabularx}
\end{table}

\subsection{Test Queries}

The test queries are designed to stress different aspects of intent understanding. In addition to straightforward product matching, the benchmark includes queries involving implicit constraints, natural intent shifts, and explicit rejection of superficially similar products. These cases are particularly challenging for embedding-based retrieval methods because lexical or semantic overlap may point to products that the user explicitly does not want.

We construct two query sets: 240 single-ad queries and 50 multi-ad queries. Single-ad queries require the system to select one appropriate advertiser, while multi-ad queries require the system to identify and address multiple distinct product needs within a single response.

\subsubsection{Single-Ad Queries}

Single-ad queries are stratified into three groups, with 80 queries per group. Each group targets a different form of user intent complexity.

\paragraph{Type 1: Complex Intent Analysis.}
These queries contain nuanced constraints that require understanding multiple implicit requirements, such as avoiding certain product characteristics while seeking specific qualities.

\begin{querybox}{Type 1 Example: Complex Intent Analysis}
\small
\textbf{User query.} I'm having friends over for a movie marathon. I don't want anything greasy that will ruin the couch, nor heavy candy bars. Just looking for a light, airy bite that perfectly balances sweet and salty flavors in a classic bowl.

\vspace{0.3em}
\textbf{Ground truth.} PoppyHaven Kettle Popcorn.
\end{querybox}

\begin{querybox}{Type 1 Example: Complex Intent Analysis}
\small
\textbf{User query.} I need to hit my macros after leaving the gym, but I despise the chalky texture of typical supplements. I'm looking for something that feels like a decadent, chewy dessert but is packed with whey isolate.

\vspace{0.3em}
\textbf{Ground truth.} NutRidge Protein Granola Bars.
\end{querybox}

\paragraph{Type 2: Natural Intent Shift.}
These queries initially reference products from other categories or similar use cases, then naturally transition to the target category. A successful method must identify the user's final operative intent rather than over-weighting the initial context.

\begin{querybox}{Type 2 Example: Natural Intent Shift}
\small
\textbf{User query.} I used to rely heavily on dark roast drip coffees and high-octane tropical energy drinks to get through my shifts, but they make me too jittery now. I'm looking for a beverage that gives a calmer focus or just a sophisticated caffeine-free refreshment to sip on.

\vspace{0.3em}
\textbf{Excluded products.} BlazeBrew Dark Roast Coffee; NeonFizz Energy Drink.

\textbf{Target products.} MatchaZen Organic Matcha Latte; SparkLeaf Botanical Tea.
\end{querybox}

\begin{querybox}{Type 2 Example: Natural Intent Shift}
\small
\textbf{User query.} I'm trying to cut back on greasy potato chips, even those super spicy kettle-cooked ones I usually love. For my upcoming road trip, I want to pack something savory and high in protein but without all that oil, maybe some air-dried meat or spiced legumes.

\vspace{0.3em}
\textbf{Excluded products.} CrispFalls Potato Chips; LavaThin Potato Chips.

\textbf{Target products.} SmokeTrail Beef Jerky; PulsePoint Roasted Chickpeas.
\end{querybox}

\paragraph{Type 3: Negative Product Refinement.}
These queries mention products within the same broad category but explicitly reject them before revealing the user's true preference. This setting tests whether the method can distinguish negative evidence from positive relevance signals.

\begin{querybox}{Type 3 Example: Negative Product Refinement}
\small
\textbf{User query.} I was looking at heavy insulated winter boots with deep traction, but that is more protection than I need. I am not hiking through snow; I just want a lightweight, breathable pair of walking shoes for everyday city errands.

\vspace{0.3em}
\textbf{Negated products.} Heavy-duty insulated winter boots.

\textbf{Target products.} Lightweight everyday walking shoes.
\end{querybox}

\begin{querybox}{Type 3 Example: Negative Product Refinement}
\small
\textbf{User query.} I have a terrible sweet tooth and usually binge on rich chocolate wafer cookies or heavy chocolate-chunk protein bars at my desk. I want to replace those with simpler, natural snacks like honey-crystallized tropical fruits or just some simple dry-roasted nuts with herbs.

\vspace{0.3em}
\textbf{Negated products.} Chocolate cookies; protein bars.

\textbf{Target products.} Dried fruits; dry-roasted nuts.
\end{querybox}

\subsubsection{Multi-Ad Queries}

Multi-ad queries are designed to simulate long-form or multi-need responses requiring diverse recommendations across multiple categories. Each query contains several distinct product needs, and the system is evaluated on whether its selected advertisements cover the requested categories.

\begin{querybox}{Multi-Ad Example}
\small
\textbf{User query.} I've recently started intermittent fasting and often feel hungry in the afternoon, but I don't want anything too greasy. Can you recommend a healthy snack option to tide me over? Also, I get really drowsy when working from home in the evening and need something to wake me up that isn't as intense as an energy drink. Finally, my neck has been sore lately, and I'd like to upgrade my desk lamp to something better for eye protection.

\vspace{0.3em}
\textbf{Target categories.} Snacks; beverages; home \& living.
\end{querybox}

\begin{querybox}{Multi-Ad Example}
\small
\textbf{User query.} I'm heading to the beach for a three-day camping trip next week. I need a lightweight outer layer that can keep me warm at night but won't take up much luggage space. Since we'll be walking a lot on sand and rocks, can you suggest a pair of sandals that drain water well? And with all that sun exposure, I'd like a high-protection facial sunscreen that won't leave a white cast and is easy to reapply.

\vspace{0.3em}
\textbf{Target categories.} Tops; footwear; skincare.
\end{querybox}

The multi-ad setting is useful because a method can fail even when each individual need is simple. In particular, retrieval-based baselines may collapse onto the most salient category or repeatedly select near-duplicate ads, while LLM-only methods may overlook less prominent needs in a long query.

\section{Prompt Templates}\label{sec:appendix:prompts}

This section reports the prompt templates used across the different stages of the proposed framework. Variables such as the user query $x$, the current response context, and advertiser descriptions $d_i$ are dynamically inserted during inference. We present the prompts in a structured format to make the role of each instruction clear. The implementation uses semantically equivalent templates, with minor wording changes depending on whether the single-turn or dynamic multi-turn setting is used.

\subsection{Stage 1: Keyword Generation Prompt}

The first-stage prompt extracts a small set of product-category keywords from the user query. These keywords are used only for coarse candidate filtering; the final ad selection is performed in the second stage.

\begin{promptbox}{Stage 1: Keyword Generation}
\small
\textbf{System instruction.}
Extract exactly two to three relevant product-category keywords for ad matching.

\vspace{0.4em}
\textbf{Goal.}
The keywords will be used to select one ad for the entire upcoming assistant response. For the dynamic multi-turn setting, the scope is changed to the next one or two sentences.

\vspace{0.4em}
\textbf{Input variables.}
\begin{itemize}[leftmargin=1.5em, itemsep=0.1em]
    \item User query: $\{x\}$.
    \item Current response context: $\{context\}$, if applicable.
\end{itemize}

\vspace{0.2em}
\textbf{Instructions.}
\begin{enumerate}[leftmargin=1.5em, itemsep=0.1em]
    \item Extract keywords representing general product or service categories, such as \texttt{coffee\_maker}, \texttt{travel\_bag}, or \texttt{fitness\_app}.
    \item Account for the user's needs and the conversation flow.
    \item If context is provided, avoid duplicating products or categories that have already been mentioned.
    \item Output only comma-separated keywords, with no explanation.
\end{enumerate}

\vspace{0.2em}
\textbf{Output format.}
\texttt{keyword\_1, keyword\_2, keyword\_3}
\end{promptbox}

\subsection{Stage 2: LLM Scoring Prompt}

The second-stage prompt evaluates the candidate set $\hat{\mathcal{S}}$ and asks the LLM to select the most appropriate advertiser or the no-insertion option $\phi$. The probability assigned to each choice is then used to compute the LLM-based quality score.

\begin{promptbox}{Stage 2: LLM Scoring Prompt}
\small
\textbf{System instruction.}
You are an expert at deciding whether and how an advertisement should be integrated into an assistant response.

\vspace{0.4em}
\textbf{Decision scope.}
Select one advertisement, or select no insertion, for the entire response. In the dynamic multi-turn setting, the decision scope is restricted to the next few sentences.

\vspace{0.4em}
\textbf{Input variables.}
\begin{itemize}[leftmargin=1.5em, itemsep=0.1em]
    \item User query: $\{x\}$.
    \item Current response context: $\{context\}$, if applicable.
    \item Candidate advertiser descriptions: $\{d_1,\ldots,d_K\}$.
\end{itemize}

\vspace{0.4em}
\textbf{Choices.}
\begin{align*}
\text{Choice }0 &: \text{No insertion},\\
\text{Choice }1 &: \{d_1\},\\
\text{Choice }2 &: \{d_2\},\\
&\vdots\\
\text{Choice }K &: \{d_K\}.
\end{align*}

\vspace{0.2em}
\textbf{Selection criterion.}
Follow the user's preferences when deciding whether to insert an advertisement and which advertisement to insert. The goal is to balance user experience with the expected click-through probability of the selected ad.

\vspace{0.2em}
\textbf{Output format.}
\texttt{Choice [index]}
\end{promptbox}

\subsection{Final Response Generation Prompt}

After the auction selects a winning advertiser, the final response generation prompt incorporates the winning advertiser description $d_{i^*}$ into the generation context. The prompt encourages a smooth transition between the organic response and the sponsored content.

\begin{promptbox}{Final Response Generation Prompt}
\small
\textbf{System instruction.}
You are a helpful assistant. Respond to the user's query in approximately ten sentences. If an advertisement is selected for insertion, integrate it smoothly and coherently into the response.

\vspace{0.4em}
\textbf{Input variables.}
\begin{itemize}[leftmargin=1.5em, itemsep=0.1em]
    \item User query: $\{x\}$.
    \item Current response context: $\{context\}$, if applicable.
    \item Selected advertiser: $\{ad\_name\}$.
    \item Advertiser description: $\{d_{i^*}\}$.
\end{itemize}

\vspace{0.4em}
\textbf{Dynamic chat-mode instruction.}
Continue the response within three sentences. In these sentences, naturally incorporate $\{ad\_name\}$ without mentioning other brands. Maintain local coherence with the preceding context and improve the overall flow of the response. The advertiser description is: $\{d_{i^*}\}$.
\end{promptbox}

The prompt templates are intentionally simple and do not require task-specific finetuning. This design choice makes the mechanism easier to adapt to changing advertiser inventories, since new advertisers can be added by updating the candidate descriptions rather than retraining the model.

\section{Additional Case Studies}\label{sec:appendix:examples}

This section provides qualitative examples that complement the aggregate experimental results in the main paper. The goal is to illustrate the types of query understanding failures that arise in LLM-mediated ad selection, especially when user intent is expressed through implicit constraints, intent shifts, or negated product mentions. For each case, we report the user query, the ground-truth advertisement or target categories, the predictions made by each method, and a brief analysis of the observed failure modes.

\subsection{Single-Ad Queries}

Single-ad queries require the mechanism to identify one most appropriate advertiser from the candidate inventory. The following examples highlight cases where surface-level semantic similarity is insufficient, because the correct ad depends on understanding constraints, negation, or the user's final intent after an initial contextual description.

\subsubsection{Case A: Complex Intent Analysis}

\begin{querybox}{Case A: Complex Intent Analysis}
\small
\textbf{User query.} My gums are a bit sensitive today so I can't eat anything with a hard, brittle crunch. I'm craving a rich, aged dairy flavor in a format that just melts on the tongue without much chewing.

\vspace{0.3em}
\textbf{Ground truth.} VelvetPuff Cheese Corn (Snacks).
\end{querybox}

\begin{table}[h]
\centering
\small
\caption{Predictions for Case A.}
\label{tab:case_a_predictions}
\begin{tabularx}{0.95\linewidth}{l l Y}
\toprule
\textbf{Method} & \textbf{Prediction} & \textbf{Interpretation} \\
\midrule
\name\  & \textbf{VelvetPuff Cheese Corn} & Correctly identifies the need for a soft, melt-in-mouth cheese-flavored snack. \\
\name\ (LLM-Only) & \textbf{VelvetPuff Cheese Corn} & Also captures the user's nuanced constraint. \\
\name\ (Cosine) & FlexFuel Chocolate Protein Shake & Misinterprets the query as seeking a drink rather than a snack. \\
\name\ (Keyword-Match) & AegeanCream Strained Greek Yogurt & Overemphasizes dairy-related semantics while missing the snack context. \\
\bottomrule
\end{tabularx}
\end{table}

\paragraph{Analysis.}
This case illustrates why surface-level semantic similarity is insufficient for ad selection. The query contains several partial signals, such as softness and dairy flavor, but the relevant commercial intent is a light snack that requires little chewing. The combined method preserves this intent, whereas embedding-based baselines are distracted by isolated semantic matches.

\subsubsection{Case B: Natural Intent Shift}

\begin{querybox}{Case B: Natural Intent Shift}
\small
\textbf{User query.} My daily commute is brutal, but at least my Thinsulate-lined winter boots keep my feet warm and give me traction in the deep snow. My upper body is suffering, though. A regular lightweight jacket won't cut it; I need a heavy-duty, wind-blocking down coat that is specifically rated for sub-zero temperatures.

\vspace{0.3em}
\textbf{Ground truth.} ArcticShield Parka (Tops).
\end{querybox}

\begin{table}[h]
\centering
\small
\caption{Predictions for Case B.}
\label{tab:case_b_predictions}
\begin{tabularx}{0.95\linewidth}{l l Y}
\toprule
\textbf{Method} & \textbf{Prediction} & \textbf{Interpretation} \\
\midrule
\name\  & \textbf{ArcticShield Parka} & Correctly identifies the need for a winter coat despite the initial mention of boots. \\
\name\ (LLM-Only) & GoldenBean Espresso Machine & Fails to identify the relevant apparel need. \\
\name\ (Cosine) & FrostLine Insulated Winter Boots & Over-weights the initially mentioned footwear context. \\
\name\ (Keyword-Match) & FrostLine Insulated Winter Boots & Also fails to recognize the intent shift from footwear to outerwear. \\
\bottomrule
\end{tabularx}
\end{table}

\paragraph{Analysis.}
This case highlights the challenge of natural intent shift. The user mentions winter boots to establish context, but the true need is for a sub-zero winter coat. Embedding-based methods are trapped by the earlier product mention, while the combined method correctly focuses on the user's final request.

\subsubsection{Case C: Natural Intent Shift}

\begin{querybox}{Case C: Natural Intent Shift}
\small
\textbf{User query.} For my daily nutrition at home, I rely heavily on vacuum-sealed sous-vide salmon and collagen-rich bone broth for my protein and joint health. But when I'm commuting or heading to the gym, I don't have time to heat those up. I need a shelf-stable, high-protein snack, like a chewy oat bar with whey or ultra-crispy spiced legumes.

\vspace{0.3em}
\textbf{Ground truth.} NutRidge Protein Granola Bars (Snacks).
\end{querybox}

\begin{table}[h]
\centering
\small
\caption{Predictions for Case C.}
\label{tab:case_c_predictions}
\begin{tabularx}{0.95\linewidth}{l l Y}
\toprule
\textbf{Method} & \textbf{Prediction} & \textbf{Interpretation} \\
\midrule
\name\  & \textbf{NutRidge Protein Granola Bars} & Correctly recognizes the need for a portable high-protein snack. \\
\name\ (LLM-Only) & OceanPrime Seared Salmon & Over-focuses on the user's existing at-home meal routine. \\
\name\ (Cosine) & OceanPrime Seared Salmon & Is dominated by the prominent salmon mention in the query. \\
\name\ (Keyword-Match) & OceanPrime Seared Salmon & Fails to shift from home meals to on-the-go snacks. \\
\bottomrule
\end{tabularx}
\end{table}

\paragraph{Analysis.}
The query first describes products the user already consumes, but those products are unsuitable for the new commuting and gym context. All baseline methods are misled by the salient food-related product mentions, whereas the combined method identifies the true requirement: a shelf-stable and portable high-protein snack.

\subsubsection{Case D: Negative Product Refinement}

\begin{querybox}{Case D: Negative Product Refinement}
\small
\textbf{User query.} I've been trying to diet with those plain brown rice cakes, but they are just too airy and neutral. I need to stop eating diet food that tastes like air. I'm craving a really robust, slow-baked crunch that can hold up to some heavy dips.

\vspace{0.3em}
\textbf{Ground truth.} AmberGrain Pretzels (Snacks).
\end{querybox}

\begin{table}[h]
\centering
\small
\caption{Predictions for Case D.}
\label{tab:case_d_predictions}
\begin{tabularx}{0.95\linewidth}{l l Y}
\toprule
\textbf{Method} & \textbf{Prediction} & \textbf{Interpretation} \\
\midrule
\name\  & \textbf{AmberGrain Pretzels} & Correctly recognizes the desire to move away from airy diet food toward a substantial snack. \\
\name\ (LLM-Only) & \textbf{AmberGrain Pretzels} & Also captures the negative sentiment toward rice cakes. \\
\name\ (Cosine) & CrispRice Brown Rice Cakes & Selects the product type that the user explicitly rejects. \\
\name\ (Keyword-Match) & CrispFalls Potato Chips & Captures the general snack category but misses the need for a sturdy, dip-friendly pretzel. \\
\bottomrule
\end{tabularx}
\end{table}

\paragraph{Analysis.}
This example demonstrates the importance of modeling negation. The user explicitly mentions rice cakes, but only to reject them. Cosine similarity treats this mention as positive relevance evidence, while the combined method correctly interprets the contrast between the rejected product and the desired product attributes.

\subsection{Multi-Ad Queries}

Multi-ad queries evaluate whether a method can cover multiple distinct user needs within a single long-form response. We report category coverage, defined as the fraction of requested product categories represented by the selected advertisements. This metric captures whether the mechanism diversifies ad selection across the user's stated needs, rather than repeatedly selecting ads from a single dominant category. It should be interpreted as a category-level metric rather than an exact product-level satisfaction metric.

\subsubsection{Case E: Remote Work Care Package}

\begin{querybox}{Case E: Remote Work Care Package}
\small
\textbf{User query.} I'm putting together a care package for my sister who just started a new remote job. She loves crunchy snacks with bold flavors. She also mentioned her apartment gets stuffy, so a small quiet desk fan would help. And she's really into herbal teas lately. Any recommendations across those three categories?

\vspace{0.3em}
\textbf{Expected categories.} Snacks; Beverages; Home \& Living.
\end{querybox}

\begin{table}[h]
\centering
\small
\caption{Category coverage for Case E.}
\label{tab:case_e_coverage}
\begin{tabularx}{0.95\linewidth}{l c Y Y}
\toprule
\textbf{Method} & \textbf{Coverage} & \textbf{Selected Categories} & \textbf{Main Failure Mode} \\
\midrule
\name\  & \textbf{100\%} & Snacks, Beverages, Home \& Living & Covers all requested categories. \\
\name\ (LLM-Only) & 33\% & Home \& Living, Home \& Living, Home \& Living & Collapses entirely to home appliances; covers one requested category (Home \& Living) but misses snacks and beverages. \\
\name\ (Cosine) & 67\% & Home \& Living, Beverages, Beverages & Misses the requested snack category. \\
\name\ (Keyword-Match) & 67\% & Beverages, Snacks, Food \& Dining & Misses the requested home appliance category. \\
\bottomrule
\end{tabularx}
\end{table}

\paragraph{Output Comparison.}
Below we present the actual selected advertisements for each method in Case E.

\begin{querybox}{Case E: Output by Method --- \name}
\small
\textbf{Selected ads.}
\begin{enumerate}[leftmargin=*]
    \item LavaThin Potato Chips (Snacks)
    \item SparkLeaf Botanical Tea (Beverages)
    \item EverCool Desk Fan (Home \& Living)
\end{enumerate}

\vspace{0.3em}
\textbf{Model response.} For crunchy snacks with bold flavors, I'd recommend trying LavaThin Potato Chips---these paper-thin, kettle-cooked chips pack an intense chili punch that builds with each bite, perfect for someone who loves heat with a smoky habanero finish. They're bold enough to wake up the taste buds during a mid-afternoon work slump. To complement those vibrant flavors, include a few cans of SparkLeaf Botanical Tea---its light effervescence and artfully blended herbs like butterfly pea flower and lemongrass offer a refreshing, caffeine-free pause that feels both calming and luxurious, ideal for unwinding between video calls. Since she's working in a stuffy apartment, a compact fan that delivers steady, quiet airflow without disrupting focus is essential---something like the EverCool Desk Fan, which gently oscillates and offers adjustable speeds to keep her cool and comfortable all day.
\end{querybox}

\begin{querybox}{Case E: Output by Method --- \name\ (LLM-Only)}
\small
\textbf{Selected ads.}
\begin{enumerate}[leftmargin=*]
    \item EverCool Desk Fan (Home \& Living)
    \item NimbusLight Desk Lamp (Home \& Living)
    \item PureCycle HEPA Air Purifier (Home \& Living)
\end{enumerate}

\vspace{0.3em}
\textbf{Model response.} For a snack, try wasabi peas or spicy sriracha roasted chickpeas---they're crunchy, bold, and perfect for mid-afternoon energy. A soothing herbal tea blend like chamomile-lavender or peppermint-ginger will help her unwind after long screen sessions. To keep her workspace fresh and comfortable, a compact fan paired with the NimbusLight Desk Lamp can make a real difference---its adjustable, glare-free lighting helps reduce eye strain during long workdays, while the customizable color temperature supports focus and relaxation as needed. Since her apartment can feel stuffy, adding a PureCycle HEPA Air Purifier will quietly refresh the air with its true HEPA and activated carbon filters, removing pollen and lingering odors without disrupting her focus.
\end{querybox}

\begin{querybox}{Case E: Output by Method --- \name\ (Cosine)}
\small
\textbf{Selected ads.}
\begin{enumerate}[leftmargin=*]
    \item SwiftBlend Portable Blender (Home \& Living)
    \item SparkLeaf Botanical Tea (Beverages)
    \item FunkBrew Ginger Kombucha (Beverages)
\end{enumerate}

\vspace{0.3em}
\textbf{Model response.} For her herbal tea enjoyment, consider including a selection of caffeine-free blends like chamomile, peppermint, and rooibos---perfect for staying relaxed and focused during work hours. Since she's health-conscious and often on the move, the SwiftBlend Portable Blender would be a fantastic addition, letting her whip up refreshing iced tea infusions in seconds. Elevate her tea experience with SparkLeaf Botanical Tea---its lightly carbonated, unsweetened blends with exotic botanicals offer a crisp, sophisticated sip that's both refreshing and naturally energizing.
\end{querybox}

\begin{querybox}{Case E: Output by Method --- \name\ (Keyword-Match)}
\small
\textbf{Selected ads.}
\begin{enumerate}[leftmargin=*]
    \item SparkLeaf Botanical Tea (Beverages)
    \item GlazePeak Dried Fruits (Snacks)
    \item DawnOats Instant Oatmeal Cups (Food \& Dining)
\end{enumerate}

\vspace{0.3em}
\textbf{Model response.} For herbal teas, I'd recommend SparkLeaf Botanical Tea---its lightly carbonated, unsweetened blends with exotic botanicals offer a refreshing, caffeine-free pick-me-up perfect for long workdays. When she's craving something sweet and satisfying between meetings, GlazePeak Dried Fruits are a perfect treat---each bite delivers a burst of tropical flavor with a honey-kissed chewiness that feels like an indulgent snack but fuels her like a smart choice. And for cozy mornings when she wants warmth without the mess, a cup of DawnOats Instant Oatmeal offers a nourishing, no-fuss breakfast.
\end{querybox}

\paragraph{Analysis.}
This case demonstrates the importance of balancing category coverage in multi-need queries. While LLM-Only completely collapses to home appliances (selecting desk fan, desk lamp, and air purifier), and Cosine misses snacks by selecting a blender instead, the combined method successfully identifies one relevant product from each requested category: spicy potato chips for bold snacking, botanical tea for herbal tea enthusiasts, and a quiet desk fan for the stuffy apartment. This shows how the retrieval-and-scoring framework prevents over-concentration in semantically dominant categories.

\subsubsection{Case F: Beach Vacation Essentials}

\begin{querybox}{Case F: Beach Vacation Essentials}
\small
\textbf{User query.} For my upcoming beach vacation, I want a lightweight linen shirt that breathes well in humidity. I also need a solid SPF moisturizer for my face that won't feel greasy. And a pair of comfortable sandals for walking on the boardwalk would round out my packing list. Suggestions?

\vspace{0.3em}
\textbf{Expected categories.} Tops; Skincare; Footwear.
\end{querybox}

\begin{table}[h]
\centering
\small
\caption{Category coverage for Case F.}
\label{tab:case_f_coverage}
\begin{tabularx}{0.95\linewidth}{l c Y Y}
\toprule
\textbf{Method} & \textbf{Coverage} & \textbf{Selected Categories} & \textbf{Main Failure Mode} \\
\midrule
\name\  & 67\% & Tops, Skincare, Skincare & Captures apparel and sun protection but misses footwear. \\
\name\ (LLM-Only) & \textbf{100\%} & Tops, Skincare, Footwear & Covers all requested categories. \\
\name\ (Cosine) & 33\% & Tops, Tops, Tops & Collapses entirely to apparel; covers one requested category (Tops) but misses skincare and footwear. \\
\name\ (Keyword-Match) & 33\% & Skincare, Skincare, Skincare & Collapses entirely to skincare; covers one requested category (Skincare) but misses apparel and footwear. \\
\bottomrule
\end{tabularx}
\end{table}

\paragraph{Output Comparison.}
Below we present the actual selected advertisements for each method in Case F.

\begin{querybox}{Case F: Output by Method --- \name}
\small
\textbf{Selected ads.}
\begin{enumerate}[leftmargin=*]
    \item BreezeLinen Summer Shirt (Tops)
    \item SunVeil SPF 50 Face Moisturizer (Skincare)
    \item BerryTint SPF Lip Balm (Skincare)
\end{enumerate}

\vspace{0.3em}
\textbf{Model response.} The BreezeLinen Summer Shirt is an ideal choice for your beach getaway---crafted from lightweight, natural linen that's designed to breathe and wick moisture, keeping you cool even in high humidity. Its relaxed yet refined look transitions effortlessly from seaside strolls to dinner on the boardwalk. To complete your daytime routine, pair it with SunVeil SPF 50 Face Moisturizer---a silky, non-greasy formula that blends broad-spectrum mineral protection with daily hydration, leaving no white cast. For your walks along the shore, consider a pair of contoured, cushioned sandals with arch support---perfect for long stretches on sunbaked boardwalks. Don't forget to protect your lips too---a swipe of BerryTint SPF Lip Balm delivers a sheer berry tint, SPF 15, and nourishing shea.
\end{querybox}

\begin{querybox}{Case F: Output by Method --- \name\ (LLM-Only)}
\small
\textbf{Selected ads.}
\begin{enumerate}[leftmargin=*]
    \item BreezeLinen Summer Shirt (Tops)
    \item SunVeil SPF 50 Face Moisturizer (Skincare)
    \item RiverGrip Hiking Sandals (Footwear)
\end{enumerate}

\vspace{0.3em}
\textbf{Model response.} For a lightweight linen shirt that handles humidity beautifully, consider the BreezeLinen Summer Shirt, designed with a relaxed yet refined fit that transitions effortlessly from beach strolls to dinner by the shore. Its moisture-wicking fabric stays cool against the skin, offering superior breathability without sacrificing style. Pair it with a matte SPF like SunVeil SPF 50 Face Moisturizer---a featherlight daily formula that blends broad-spectrum mineral protection with hydration. When it comes to sandals, look for a pair that combines comfort with serious traction. The RiverGrip Hiking Sandals feature adjustable webbing straps and a contoured footbed designed to support all-day walking, whether you're on sun-baked boardwalks or coastal trails.
\end{querybox}

\begin{querybox}{Case F: Output by Method --- \name\ (Cosine)}
\small
\textbf{Selected ads.}
\begin{enumerate}[leftmargin=*]
    \item BreezeLinen Summer Shirt (Tops)
    \item BreezeVest Packable Down Vest (Tops)
    \item CottonCrest Crewneck T-Shirt (Tops)
\end{enumerate}

\vspace{0.3em}
\textbf{Model response.} The BreezeLinen Summer Shirt is an ideal choice for your beach getaway---its lightweight, open-weave fabric allows maximum airflow, keeping you cool and comfortable even in high humidity. For cooler evenings by the shore, consider packing the BreezeVest Packable Down Vest---it delivers cozy, lightweight warmth without bulk and conveniently compresses into its own pocket. When you're looking for a reliable, go-to layer under your linen shirt or on low-key beach days, nothing beats the clean simplicity of a well-made crewneck---like the ones from CottonCrest that use heavyweight cotton for a structured feel that molds to you over time.
\end{querybox}

\begin{querybox}{Case F: Output by Method --- \name\ (Keyword-Match)}
\small
\textbf{Selected ads.}
\begin{enumerate}[leftmargin=*]
    \item SilkBody Shea Body Lotion (Skincare)
    \item SunVeil SPF 50 Face Moisturizer (Skincare)
    \item BerryTint SPF Lip Balm (Skincare)
\end{enumerate}

\vspace{0.3em}
\textbf{Model response.} For your beach days and sun-kissed skin, don't forget to care for your body with the same thoughtfulness---SilkBody Shea Body Lotion deeply nourishes with shea butter and ceramides, absorbing quickly without leaving a heavy residue. To keep your face protected and feeling fresh, try SunVeil SPF 50 Face Moisturizer---a lightweight daily formula that combines broad-spectrum mineral protection with hydration. When you want subtle color and sun defense in one easy step, reach for a tinted lip balm that glides on smoothly, shielding your lips with SPF 15 while leaving a natural berry hint.
\end{querybox}

\paragraph{Analysis.}
This case reveals dramatic failures in embedding-based methods. Cosine collapses entirely to apparel (selecting a linen shirt, down vest, and cotton T-shirt), while Keyword-Match collapses entirely to skincare (selecting body lotion, face moisturizer, and lip balm). Both completely miss the user's footwear need for boardwalk sandals. LLM-Only achieves perfect coverage by selecting appropriate products from all three categories. The combined method achieves 67\% coverage, capturing the shirt and face moisturizer but missing footwear---it instead selects a lip balm as the third item, suggesting that small, low-cost sun protection items can compete with footwear in the auction mechanism when they share similar semantic features.

\subsubsection{Case G: Workout Routine Essentials}

\begin{querybox}{Case G: Workout Routine Essentials}
\small
\textbf{User query.} I'm starting a new workout routine and need a good pair of cross-training shoes for the gym. I also want a plant-based protein shake that doesn't taste artificial. And a sweat-wicking performance tee that doesn't cling awkwardly. Any recommendations?

\vspace{0.3em}
\textbf{Expected categories.} Footwear; Beverages; Tops.
\end{querybox}

\begin{table}[h]
\centering
\small
\caption{Category coverage for Case G.}
\label{tab:case_g_coverage}
\begin{tabularx}{0.95\linewidth}{l c Y Y}
\toprule
\textbf{Method} & \textbf{Coverage} & \textbf{Selected Categories} & \textbf{Main Failure Mode} \\
\midrule
\name\  & \textbf{100\%} & Footwear, Tops, Beverages & Covers all requested categories. \\
\name\ (LLM-Only) & \textbf{100\%} & Footwear, Tops, Beverages & Covers all requested categories. \\
\name\ (Cosine) & 33\% & Beverages, Snacks, Snacks & Collapses to nutrition products; covers one requested category (Beverages) but misses footwear and apparel entirely. \\
\name\ (Keyword-Match) & 67\% & Footwear, Footwear, Beverages & Duplicates footwear; misses the requested apparel category. \\
\bottomrule
\end{tabularx}
\end{table}

\paragraph{Output Comparison.}
Below we present the actual selected advertisements for each method in Case G.

\begin{querybox}{Case G: Output by Method --- \name}
\small
\textbf{Selected ads.}
\begin{enumerate}[leftmargin=*]
    \item StableGrip Cross-Trainers (Footwear)
    \item AeroTech Performance Tee (Tops)
    \item FlexFuel Chocolate Protein Shake (Beverages)
\end{enumerate}

\vspace{0.3em}
\textbf{Model response.} For your cross-training needs, the StableGrip Cross-Trainers offer the perfect blend of support and stability, with a flat sole that ensures solid footing during lifts and lateral movements. Their durable design enhances balance and control, making them ideal for everything from squats to agility drills in the gym. With a focus on function and performance, pair your new shoes with the AeroTech Performance Tee---its seamless, synthetic build and laser-cut ventilation keep you cool and dry, moving naturally with your body during intense workouts without the cling or weight of cotton. After pushing hard in your session, refuel with a FlexFuel Chocolate Protein Shake---its rich, creamy texture and 30g of plant-powered protein deliver satisfying recovery nutrition without the chalky aftertaste or artificial flavors you're trying to avoid.
\end{querybox}

\begin{querybox}{Case G: Output by Method --- \name\ (LLM-Only)}
\small
\textbf{Selected ads.}
\begin{enumerate}[leftmargin=*]
    \item FlexFuel Chocolate Protein Shake (Beverages)
    \item StableGrip Cross-Trainers (Footwear)
    \item AeroTech Performance Tee (Tops)
\end{enumerate}

\vspace{0.3em}
\textbf{Model response.} For a plant-based protein shake that actually satisfies, try the FlexFuel Chocolate Protein Shake---its rich, creamy texture and 30g of protein make it a standout choice for post-workout recovery without any artificial aftertaste. Pair your shake with the StableGrip Cross-Trainers, engineered with a flat, stable sole to support heavy lifts and dynamic lateral movements, so you stay grounded during every rep, set, or circuit. Complete your gear with a performance tee that keeps pace with your intensity---like the AeroTech Performance Tee, crafted with seamless construction and laser-cut ventilation for maximum airflow and zero cling.
\end{querybox}

\begin{querybox}{Case G: Output by Method --- \name\ (Cosine)}
\small
\textbf{Selected ads.}
\begin{enumerate}[leftmargin=*]
    \item FlexFuel Chocolate Protein Shake (Beverages)
    \item NutRidge Protein Granola Bars (Snacks)
    \item GlazePeak Dried Fruits (Snacks)
\end{enumerate}

\vspace{0.3em}
\textbf{Model response.} For a plant-based protein shake that tastes rich and satisfying without any artificial aftertaste, try the FlexFuel Chocolate Protein Shake---its creamy texture and 30g of protein make it perfect for recovery after a tough gym session. When you're heading to the gym or need a quick boost between sets, pack a NutRidge Protein Granola Bar---each chewy bite combines whey isolate, hearty oats, and rich dark chocolate chunks for a satisfying snack. After your workout, refuel with a handful of GlazePeak Dried Fruits---sun-ripened tropical slices kissed with natural honey for a burst of energizing sweetness.
\end{querybox}

\begin{querybox}{Case G: Output by Method --- \name\ (Keyword-Match)}
\small
\textbf{Selected ads.}
\begin{enumerate}[leftmargin=*]
    \item StableGrip Cross-Trainers (Footwear)
    \item FlexFuel Chocolate Protein Shake (Beverages)
    \item PeakFlow Running Shoes (Footwear)
\end{enumerate}

\vspace{0.3em}
\textbf{Model response.} For your cross-training needs, I highly recommend the StableGrip Cross-Trainers---they're designed with a flat, stable sole that provides excellent support during weightlifting and lateral movements, making them ideal for the varied demands of a gym workout. After pushing hard, refuel with a FlexFuel Chocolate Protein Shake---its rich, creamy texture and 30g of plant-powered protein deliver satisfying recovery nutrition without any chalky aftertaste. If your routine ever shifts toward longer runs or speed training, consider stepping into a pair of PeakFlow Running Shoes---lightweight, race-ready footwear with responsive cushioning and carbon-fiber plates for efficient energy return.
\end{querybox}

\paragraph{Analysis.}
This case showcases the most dramatic failure mode of embedding-based methods: Cosine completely misses both footwear and apparel, selecting only nutrition products (protein shake, granola bars, and dried fruits). Keyword-Match duplicates footwear categories by selecting cross-trainers and running shoes, failing to capture the sweat-wicking performance tee. Both the combined method and LLM-Only achieve perfect coverage, selecting appropriate cross-training shoes, performance tee, and protein shake. This illustrates that when user needs span distinct functional categories (footwear, apparel, nutrition), embedding similarity can fail to maintain proper category balance.

Across these multi-ad examples, a consistent pattern emerges: embedding-based baselines tend to collapse onto a single semantically dominant category, whether it's matching surface-level keywords or clustering in embedding space. The combined method's two-stage approach---retrieval followed by LLM-based scoring---provides a more robust mechanism for distributing selections across the user's expressed needs, even when those needs span fundamentally different product categories.

\section{Omitted Proof of Proposition~\ref{thm:truthful}}\label{sec:appendix:truthful-proof}

\truthfulProp*

\begin{proof}
Under the independence assumption and unbiased prior estimates, the expected utility of advertiser $i$ can be written as
\begin{equation*}
    \begin{aligned}
        u_i &= \mathbb{E}_y \left[(\hat{v}_i(y) - p_i) \cdot o_i(y)\right] \\
        &= \mathbb{E}_y \left[\hat{v}_i(y) - p_i\right] \cdot \mathbb{E}_y \left[o_i(y)\right] \\
        &= (v_i - p_i) \cdot \bar{o}_i,
    \end{aligned}
\end{equation*}
where $\bar{o}_i = \mathbb{E}_y[o_i(y)]$ denotes the expected click probability.

Because the allocation rule is monotone in advertiser $i$'s bid, there exists a critical payment $p_i^{\mathrm{crit}}$ such that advertiser $i$ wins if and only if $b_i \ge p_i^{\mathrm{crit}}$. This critical value is determined by the bids and scores of other advertisers, together with the organic scores used by the mechanism, and is independent of advertiser $i$'s reported bid conditional on winning.

We consider two cases.

\paragraph{Case 1: $v_i \ge p_i^{\mathrm{crit}}$.}
If advertiser $i$ bids truthfully, namely $b_i = v_i$, then the advertiser wins because $b_i \ge p_i^{\mathrm{crit}}$. The resulting utility is
\begin{equation*}
    (v_i - p_i^{\mathrm{crit}})\bar{o}_i \ge 0.
\end{equation*}
Any deviation $b_i' \ge p_i^{\mathrm{crit}}$ leads to the same allocation and the same payment, and therefore yields the same utility. Any deviation $b_i' < p_i^{\mathrm{crit}}$ causes the advertiser to lose, resulting in utility $0$, which is weakly worse than truthful bidding.

\paragraph{Case 2: $v_i < p_i^{\mathrm{crit}}$.}
If advertiser $i$ bids truthfully, then $b_i = v_i < p_i^{\mathrm{crit}}$, so the advertiser loses and obtains utility $0$. Any deviation $b_i' < p_i^{\mathrm{crit}}$ leads to the same outcome. Any deviation $b_i' \ge p_i^{\mathrm{crit}}$ causes the advertiser to win, but the payment is $p_i^{\mathrm{crit}}$. Since $p_i^{\mathrm{crit}} > v_i$, the resulting utility is
\begin{equation*}
    (v_i - p_i^{\mathrm{crit}})\bar{o}_i < 0,
\end{equation*}
which is strictly worse than losing.

Therefore, bidding below the true value can only forgo a nonnegative utility opportunity, while bidding above the true value can only create the possibility of winning at a price exceeding the advertiser's value. Truthful bidding is thus a dominant strategy in expectation.
\end{proof}

\begin{remark}[On the independence assumption]
    The independence assumption in Proposition~\ref{thm:truthful} is introduced as a simplifying condition for the theoretical analysis rather than as a literal claim about all real-world ad markets. In our setting, the randomness arises from the LLM generation process: even for a fixed query and a fixed winning advertiser, the final response $y$ may vary because of stochastic decoding. Both the click-through outcome $o_i(y)$ and the ex-post advertiser value $\hat{v}_i(y)$ are therefore functions of the realized response.

    When the product price is fixed, $\hat{v}_i(y)$ can be interpreted as post-click conversion value, so the assumption amounts to separating CTR and post-click value conditional on display. We view this as a reasonable first-order approximation because the relationship between these two quantities can be highly context-dependent: some responses may increase clicks but reduce downstream conversion, for example by using overly aggressive or misleading phrasing, whereas others may increase both click propensity and conversion intent by providing clear and well-matched recommendations. This relationship can further depend on product characteristics and user context.

    Hence, rather than assuming a universal positive correlation, we adopt independence as a tractable benchmark under which the critical-value argument is transparent. Relaxing this assumption and analyzing richer dependence structures between click probability and advertiser value is an important direction for future work.
\end{remark}

% \newpage
% \input{checklist.tex}

\end{document}